\documentclass[11pt,reqno]{amsart}
\usepackage{amsfonts}
\usepackage{amsmath}
\usepackage{amssymb}
\usepackage{amsthm}
\usepackage{graphicx}
\graphicspath{ {Matlab/} }
\usepackage{epstopdf}
\usepackage{color}
\usepackage{bm}
\usepackage{enumerate}
\usepackage[normalem]{ulem}
\usepackage[margin=1.1in]{geometry}

\def\E{{\mathbb E}}
\def\R{{\mathbb R}}

\def\PP{{\mathbb P}}

\def\Z{{\mathcal Z}}

\def\A{{\mathcal A}}
\def\F{{\mathcal F}}
\def\FF{{\mathbb F}}

\newtheorem{theorem}{Theorem}
\newtheorem{corollary}[theorem]{Corollary}
\newtheorem{definition}[theorem]{Definition}

\theoremstyle{definition}

\newtheorem{remark}[theorem]{Remark}

\title{Mean field and n-agent games for optimal investment under relative performance criteria}
\author{Daniel Lacker$^*$}
\thanks{$^*$Industrial Engineering \& Operations Research, Columbia University. Partially supported by the National Science Foundation under Award No. DMS 1502980; daniel.lacker@columbia.edu.}
\author{Thaleia Zariphopoulou$^\dagger$}
\thanks{$^\dagger$Dpts. of Mathematics and IROM, The University of Texas at Austin and the Mathematical Institute, University of Oxford; zariphop@math.utexas.edu.}
\date{}

\begin{document}

\begin{abstract}
We analyze a family of portfolio management problems under relative performance criteria, for fund managers having CARA or CRRA utilities and trading in a common investment horizon in log-normal markets. We construct explicit constant equilibrium strategies for both the finite population games and the corresponding mean field games, which we show are unique in the class of constant equilibria. In the CARA case, competition drives agents to invest more in the risky asset than they would otherwise, while in the CRRA case competitive agents may over- or under-invest, depending on their levels of risk tolerance.
\end{abstract}

\maketitle

\section{Introduction}

This paper is a contribution to both the theory of finite population and mean
field games and to optimal portfolio management under competition and
relative performance criteria. For the former, we construct explicit
solutions for both $n$-player and mean field games, providing a new family of
tractable solutions. For the latter, we formulate a new class of competition
and relative performance optimal investment problems for agents having
exponential (CARA) and power (CRRA) utilities, for both a finite number
and a continuum of agents.

The finite-population case consists of $n$ \emph{fund managers} (or \emph{agents}) trading between a
common riskless bond and an individual stock. The price of each stock is
modeled as a log-normal process driven by two independent Brownian motions. 
The first Brownian motion is the same for all prices, representing a common market noise,
while the second is idiosyncratic, specific to each individual stock.
Precisely, the $i^{\text{th}}$ fund specializes in stock $i$ whose price $(S_t^i)_{t \ge 0}$ is given by 
\begin{equation}
\frac{dS_{t}^{i}}{S_{t}^{i}}=\mu _{i}dt+\nu _{i}dW_{t}^{i}+\sigma _{i}dB_{t},
\label{stock}
\end{equation}
with constant market parameters $\mu _{i}>0$, $\sigma_{i}\ge 0$, and $\nu_{i}\ge 0$, with $\sigma_i+\nu_i>0$.
The (one-dimensional) standard Brownian motions $B,W^1,W^2,\ldots,W^n$ are independent. When $\sigma_i > 0$, the process $B$ induces a correlation between the stocks, and thus we call $B$ the \emph{common noise} and $W^{i}$ an \emph{idiosyncratic noise}.

Our setup covers the important special case in which all funds trade in the same stock; that is, $\mu_i=\mu$, $\nu_i=0$, and $\sigma_i=\sigma$ for all $i=1,\ldots,n$, for some $\mu,\sigma > 0$ independent of $i$.
In this setting, all stocks are identical and the model is to be interpreted as $n$ agents investing in a single common stock. These $n$ agents differ in their risk preferences but otherwise face the same market opportunities.
We choose to work with one-dimensional stocks $S^i$ for mere simplicity, but our analysis would adapt with purely notational changes to cover the case where each $S^i$ is a vector of stocks available to agent $i$. In this multidimensional setting, the ``single stock'' case would model the realistic situation of a large number of agents trading in the same vector of stocks.

All fund managers share a common time horizon, $T > 0$, and aim to
maximize their expected utility at $T$. The utility functions $U_1,\ldots,U_n$ are agent-specific functions of both terminal wealth, $X_{T}^{i}$, and a ``competition component,'' $\overline{X}_{T}$, which depends on
the terminal wealths of all agents. We study two representative cases,
related to the popular exponential and power utilities. 

For the exponential case, we assume that competition affects the wealth
additively and is modeled through the arithmetic average wealth of all agents,
\begin{equation}
U_{i}\left( X_{T}^{i},\overline{X}_{T}\right) =-e^{-\frac{1}{\delta _{i}}\left(X_{T}^{i}-\theta _{i}\overline{X}_{T}\right) }, \quad \text{ where } \quad  \overline{X}_{T}=\frac{1}{n}\sum_{k=1}^{n}X_{T}^{k}.  \label{expo-intro}
\end{equation}
The parameters $\delta_{i}>0$ and $\theta_{i}\in [0,1]$
represent the $i^{\text{th}}$ agent's absolute risk tolerance and absolute
competition weight, with small (resp. high) values of $\theta_{i}$
denoting low (resp. high) relative performance concern. This model is similar to and largely inspired by that of Espinosa and Touzi \cite{espinosa-touzi}.

For the power case, the competition affects the wealth multiplicatively and
is modeled through the geometric average wealth of all agents,
\begin{equation}
U_{i}\left( X_{T}^{i},\overline{X}_{T}\right) =\frac{1}{1-1/\delta_i}\left(
X_{T}^{i} \overline{X}_{T}^{-\theta_i}\right) ^{1-1/\delta _{i}}, \quad \text{ where } \quad  
\overline{X}_{T}=\left( \prod_{k=1}^{n}X_{T}^{k}\right) ^{1/n}.
\label{power-intro}
\end{equation}
Now, the parameters $\delta_{i}>0$ and $\theta_{i}\in [0,1]$ represent the $i^{\text{th}}$ agent's relative risk tolerance and relative competition weight. 
The geometric mean is used here largely for its tractability, but it also admits a natural interpretation: Quantities of the form $e^{r_1},\ldots,e^{r_n}$ have geometric mean $\exp(\frac{1}{n}\sum_{i=1}^nr_i)$, which indicates that the geometric mean of wealths is simply the exponential of the arithmetic mean of returns. In this sense, the agents are using \emph{returns} rather than \emph{absolute wealth} in measuring relative performance.

The aim is to identify Nash equilibria, namely, to find investment strategies $(\pi_{t}^{1,*},\ldots,\pi_{t}^{n,*})_{t \in [0,T]}$ such that $\pi_{t}^{i,*}$ is the optimal stock allocation exercised by the $i^{\text{th}}$ agent in response to the strategy choices of all other competitors, for $i=1,\ldots,n$. As is usually the case for exponential and power
risk preferences, $\pi_{t}^{i,*}$ is taken to be the absolute wealth
and the fraction of wealth invested in the $i^{\text{th}}$ stock, respectively.
The values $\pi_{t}^{i,*}$ may be negative, indicating that the agent shorts the stock.

Competition among fund managers is well documented in investment practice
for both mutual and hedge funds; see, for example, \cite{Agarwal03,basak2015competition,brown2001careers,chevalier1997risk,
ding2008investor,Gallaher06,kempf2008tournaments,li2006consequences,sirri1998costly}.
As it is argued in these works, competition can stem, for example, from career advancement motives, seeking higher money
inflows from their clients, preferential compensation contracts.
In most of these works, only the case of two managers has been considered
and in discrete time (or two period) models, with variations of criteria
involving risk neutrality, relative performance with respect to an absolute
benchmark or a critical threshold, or constraints on the managers' risk
aversion parameters. More recently, the authors in \cite{basak2015competition} proposed a continuous time log-normal model for two fund managers with power utilities.

Asset specialization is also well documented in the finance literature,
starting with Brennan \cite{brennan1975optimal,Coval-Moskowitz,Merton87}. Other representative works include \cite{boyle2012keynes,kacperczyk2005industry,liu2014solvency,mitton2007equilibrium,
van2009information,van2010information,uppal2003model}.
As it is argued in these works, a variety of factors prompt managers to
specialize in individual stocks or asset classes, such as
familiarity, learning cost reduction, ambiguity aversion, solvency
requirements, trading costs and constraints, liquidation risks,
and informational frictions.

For tractability, we search only for Nash equilibria in which the investment strategies are constants (i.e., chosen at time zero). This restriction is quite natural, given the log-normality of the
stock prices, the scaling properties of the CARA and CRRA utilities, and
the form of the associated competition components.
To construct such an equilibrium, we first solve each single agent's optimization problem given an arbitrary (but fixed) choice of competitors' constant strategies.

Incorporating the competition component $\overline{X}$ as an additional uncontrolled state process leads to a single Hamilton-Jacobi-Bellman (HJB) equation, which we show has a unique separable smooth solution. Together with the first
order conditions, this yields the candidate policies in a closed-form. We then
construct the equilibrium through a set of compatibility conditions,
which also provide criteria for existence and uniqueness. As an intermediate step, we use arguments from indifference valuation to obtain verification results for these smooth solutions. Specifically, we interpret each HJB equation as the one solved by the writer of an individual liability determined by the competition component.

The unique constant Nash equilibrium in each model turns out to be the sum of two components. The
first is the traditional Merton portfolio (see \cite{merton1971optimum}), which is optimal for the
individual expected utility problem without any relative performance
concerns. The second component depends on the individual competition
parameter and on other quantities involving the risk tolerance and
competition parameters of all agents as well as the market parameters of
all stocks. Naturally, this second component disappears when there is no competition.

In the exponential model, it turns out that competition always results in
higher investment in the risky asset. This is not, however, the case for the power model,
mainly because the sign of the second component might not be always fixed. This sign depends on the value of the relative risk tolerance, particularly whether it is larger or smaller than one; this is to be expected given well known properties of CRRA utilities and their optimal portolios (see, for example, the so called ``nirvana'' cases in \cite{kim-omberg}).

In the noteworthy special case of a single stock, common to all agents, the equilibrium strategies are simpler.
For both the exponential and the power cases, the Nash equilibrium is of Merton type but with a modified risk tolerance, which depends linearly on the individual risk tolerance and competition parameters, with the coefficients of this linear function depending on the population averages of these parameters.

The expressions for the equilibrium strategies simplify when the number of agents $n$ tends to infinity. The limiting expressions depend solely on the limit of the empirical distribution of the \emph{type vectors} $\zeta_i=(x^i_0,\delta_i,\theta_i,\mu_i,\nu_i,\sigma_i)$, for $i=1,\ldots,n$. We show that these limiting strategies can be derived intrinsically, as equilibria of suitable mean field games (MFGs). 
Intuitively, the finite set of agents becomes a continuum, with each individual trading between the common bond and her individual stock while also competing with the rest of the (infinite) population through a relative performance functional affecting her expected terminal utility.

Although explicit solutions are available for our $n$-agent games, the MFG framework is worth introducing in this context in part because it extends naturally to more complex models, such as those involving portfolio constraints or general utility functions.  In such models, we expect the MFG framework to be more tractable than the $n$-agent games.
For instance, \cite{bielagk2017equilibrium,espinosa-touzi,frei2011financial} study $n$-agent models similar to our CARA utility model but notably including equiibrium pricing and portfolio constraints, leading to difficult $n$-dimensional quadratic BSDE systems. A MFG formulation would likely be more tractable, at least reducing the dimensionality of the problem, though we do not tackle such an analysis in this paper.

The MFG is defined in terms of a representative agent who is assigned a \emph{random} type vector $\zeta=(\xi,\delta,\theta,\mu,\nu,\sigma)$ at time zero, determining her initial wealth $\xi$, preference parameters $(\delta,\theta)$, and market parameters $(\mu,\nu,\sigma)$. The randomness of the type vector encodes the distribution of the (continuum of) agents' types.

For the exponential case, the MFG problem is to find a pair $\left( \pi ^{* },\overline{X}\right) $ with the following properties. The
investment strategy $\pi ^{* }$ optimizes, in analogy to \eqref{expo-intro},  
\begin{equation}
\sup_{\pi }\E\left[ -e^{-\frac{1}{\delta }\left( X_{T}-\overline{X}\right) }\right] ,  \label{expo-MFG}
\end{equation}
where $X_{T}$ is the wealth of the representative agent and $\overline{X}$ the
average wealth of the continuum of agents. Furthermore, at this optimum,
the consistency condition $\overline{X}=\E[X_{T}^* | \F_T^B] $ must hold, where $(\F^B_t)_{t \in [0,T]}$ is the filtration generated by the common noise $B$, and $X^*$ is the optimal wealth determined by $\pi^*$.

For the power case, the aggregate wealth $\overline{X}$ must be consistent with its $n$-agent form in \eqref{power-intro}.
With this in mind, note that the geometric mean of a positive random variable $Y$ can be written as $\exp\E[\log Y]$, whether or not the distribution of $Y$ is discrete.
This points to the MFG problem of finding a pair $\left( \pi
^{\ast },\overline{X}\right) ,$ such that $\pi ^{* }$ optimizes
\begin{equation}
\sup_{\pi }\E\left[ \frac{1}{1-1/\delta }\left( X_{T}\overline{X}^{-\theta}\right)
^{1-1/\delta }\right] ,  \label{power-MFG}
\end{equation}
and, furthermore, the consistency condition $\overline{X}=\exp \E[\log X_{T}^{*} | \F_{T}^{B}] $ holds. While its use in mean field game theory appears to be new, this notion of geometric mean of a measure is essentially a special case of the well-studied concept of \emph{generalized mean} (see \cite[Chapter III]{hardy1952inequalities}).

As in the finite population CARA and CRRA cases, we focus on the tractable
class of equilibria in which the strategy $\pi^*$ is constant in time. Such strategies are still \emph{random}, 
measurable with respect to the (time-zero-measurable) random type vector.
We solve the MFG problems directly, constructing equilibria which agree with the limiting expressions from the $n$-agent games. In each model, the solution technique is analogous to the $n$-agent setting in that we treat the aggregate wealth term as an uncontrolled state process, find a smooth separable solution of a single HJB equation, and then enforce the consistency condition.
The resulting MFG strategies take similar but notably simpler forms than their $n$-agent counterparts and exhibit the same qualitative behavior and two-component structure discussed above.

Mean field games, first introduced in \cite{lasry-lions} and \cite{huang2006large}, have by now found numerous applications in economics and finance, notably including models of income inequality \cite{gabaix2016dynamics}, economic growth \cite{lucas2014knowledge}, limit order book formation \cite{gayduk2016endogenous}, systemic risk \cite{carmona-fouque-sun}, optimal execution \cite{jaimungal2015mean,cardaliaguet-lehalle}, and oligopoly market models \cite{chan2015bertrand}, to name but a few. The closest works to ours are the static model of \cite[Section 6]{gueant-lasry-lions}, which is a competitive variant of the Markowitz model, and the stochastic growth model of \cite{huang2016mean} which has some mathematical features in common with our power utility model.
That said, our work appears to be the first application of MFG theory to portfolio optimization.

Our results add two new examples of \emph{explicitly} solvable MFG models. Beyond the linear quadratic models of \cite{bensoussan2016linear,carmona-delarue-lachapelle,carmona-fouque-sun}, such examples are scarce, especially in the presence of common noise. The only other examples we know of are those in \cite[Sections 5 and 7]{gueant-lasry-lions} as well as the more recent \cite{sun2016systemic}, which is linear-quadratic aside from a square root diffusion term. 
In fact, our models permit an explicit solution of the so-called \emph{master equation} (cf. \cite{carmona2014master}).
Moreover, we wish to emphasize the manner in which we incorporate different \emph{types} of agents, by randomizing $\zeta=(\xi,\delta,\theta,\mu,\nu,\sigma)$ as described above. Several previous works on MFGs (e.g., \cite{huang2006large}) incorporated finitely many types  by tracking a vector of mean field interactions, one for each type, but our approach has the advantage of seamlessly incorporating (uncountably) infinitely many types. While randomizing types is a standard technique in static games with a continuum of agents, the idea has scarcely appeared in the (dynamic) MFG literature; to the best of our knowledge, it has appeared only in \cite{cardaliaguet-lehalle}.

The paper is organized as follows. In Section \ref{se:exp}, we present the exponential
model and study both the $n$-agent game and the MFG. In Section \ref{se:power},
we present the analogous results for power and logarithmic utilities. For
both classes, we provide qualitative comments on the Nash and mean field equilibria, in Sections \ref{se:discussion-exp} and \ref{se:discussion-power}, respectively.
We conclude in Section \ref{se:conclusions} with a discussion of open questions and future research directions.

\section{CARA risk preferences} \label{se:exp}

We consider \emph{fund managers} (henceforth, \emph{agents}) with exponential risk preferences with constant
individual (absolute) risk tolerances. Agents are also concerned with
how their performance is measured in relation to the performances of their competitors.
This is modeled as an additive penalty term depending on the average
wealth, and weighted by an investor specific comparison parameter.

We begin our analysis with the exponential class because of its additive
scaling properties, which allow for substantial tractability. Furthermore,
the exponential class provides a direct connection with indifference
valuation, used in solving the underlying HJB equation.

\subsection{The $n$-agent game} \label{se:np-exp}

We introduce a game of $n$ agents who trade in a common investment
horizon $[0,T]$. Each agent trades between an
``individual'' stock and a riskless bond. The latter is common to all agents,
serves as the numeraire and offers zero interest rate.

Stock prices are taken to be log-normal, as described in the introduction, each driven by two independent Brownian motions. Precisely, the price $(S_t^i)_{t \in [0,T]}$ of the stock $i$ traded by the $i^{\text{th}}$  agent solves \eqref{stock}, with given market parameters $\mu _{i}>0$, $\sigma_{i}\ge 0$, and $\nu_{i}\ge 0$, with $\sigma_i+\nu_i>0$.
The independent Brownian motions $B,W^1,\ldots,W^n$ are defined on a probability space $(\Omega ,\F,\PP)$, which we endow with the natural filtration $(\F_t)_{t\in [0,T]}$ generated by these $n+1$ Brownian motions. Recall that the \emph{single stock case} is when
\begin{equation*}
(\mu _{i},\sigma _{i})=(\mu,\sigma ), \ \text{ and }\nu
_i = 0, \ \text{ for } i=1,\ldots,n,
\end{equation*}
for some $\mu,\sigma > 0$ independent of $i$. Notably, the single stock case was studied in \cite{espinosa-touzi} and \cite{frei2011financial} in greater generality, incorporating portfolio constraints and more general stock price dynamics.

Each agent $i=1,\ldots,n$ trades using a self-financing strategy, $(\pi _{t}^{i})_{t \in [0,T]}$, which represents the
(discounted by the bond) amount invested in the $i^{\text{th}}$ stock. The $i^{\text{th}}$ agent's wealth $(X^i_t)_{t \in [0,T]}$ then solves 
\begin{equation}
dX_{t}^{i} = \pi_t^i(\mu _{i}dt+\nu _{i}dW_{t}^{i}+\sigma _{i}dB_{t}),
\label{def:X-i}
\end{equation}
with $X_0^i=x_0^i \in \R$.
A portfolio strategy is deemed admissible if it belongs to the set $\A$, which consists of self-financing $\FF$-progressively measurable real-valued processes $(\pi_t)_{t \in [0,T]}$ satisfying $\E\int_0^T|\pi_t|^2dt < \infty$.

The $i^{\text{th}}$ agent's utility is a function $U_i : \R^2 \rightarrow \R$ of \textit{both} her
individual wealth, $x$, and the average wealth of all agents, $m$. It is of the form
\[
U_i(x,m) := -\exp\left(-\frac{1}{\delta_i}\left(x - \theta_im\right)\right).
\]
We will refer to the constants $\delta_i > 0$ and $\theta_i \in [0,1]$ as the \textit{personal risk tolerance} and \textit{competition weight} parameters, respectively.\footnote{
Note that $(\delta_{i},\theta_{i})$, $i=1,\ldots,n,$ are
unitless, because all wealth processes are discounted by the riskless bond.}
If agents $i=1,\ldots,n$ choose admissible strategies $\pi^1,\ldots,\pi^n$, the payoff for agent $i$ is given by 
\begin{equation}
J_i(\pi^1,\ldots ,\pi^n) := \E\left[ -\exp \left( -\frac{1}{\delta_i}\left( X_{T}^{i}-\theta_i\overline{X}_{T}\right) \right) \right]
,\quad \text{with}\quad \overline{X}_{T}=\frac{1}{n}\sum_{k=1}^{n}X_{T}^{k},
\label{J-expo}
\end{equation}
where the dynamics of $(X^i_t)_{t \in [0,T]}$ are as in \eqref{def:X-i}.
Alternatively, we may express the above as 
\begin{equation*}
J_i(\pi ^{1},\ldots ,\pi ^{n}) = \E\left[ -\exp \left( -\frac{1}{\delta_i}\left( (1-\theta_i)X^i_T + \theta_i(X_T^i-\overline{X}_{T})\right) \right) \right] ,
\end{equation*}
which highlights how the competition weight $\theta _i$ determines the $i^{\text{th}}$ agent's risk preference for \emph{absolute wealth} versus \emph{relative wealth}.
An agent with large $\theta _{i}$ (close to one) is thus more concerned with relative wealth than absolute wealth.

These interdependent optimization problems are resolved competitively, applying the concept of  Nash equilibrium in the above investment setting.

\begin{definition} \label{def:Nash}
A vector $(\pi^{1,*},\ldots,\pi^{n,*})$ of admissible strategies is a (Nash) equilibrium if, for all $\pi^i\in \A$ and $i=1,\ldots,n$, 
\begin{equation}
J_i(\pi ^{1,*},\ldots,\pi ^{i,*},\ldots,\pi ^{n,*})\geq J_i(\pi ^{1,*},\ldots
,\pi ^{i-1,*},\pi^i,\pi^{i+1,*},\ldots ,\pi^{n,*}).
\label{Nash}
\end{equation}
A constant (Nash) equilibrium is one in which, for each $i$, $\pi^{i,*}$ is constant in time, i.e., $\pi^{i,*}_t=\pi^{i,*}_0$ for all $t \in [0,T]$.\footnote{Our notion of Nash equilibrium is more accurately known as an \emph{open-loop} Nash equilibrium. A popular alternative is \emph{closed-loop} Nash equilibrium, in which agents choose strategies in terms of feedback functions as opposed to stochastic processes. However, for \emph{constant} strategies, the open-loop and closed-loop concepts coincide. That is, a constant (open-loop) Nash equilibrium is also a closed-loop Nash equilibrium, and vice versa.}
\end{definition}

\begin{remark}
Because the filtration $\FF$ is Brownian, it holds for any admissible strategy $\pi \in \A$ that $\pi_0$ is nonrandom. With this in mind, a constant Nash equilibrium will be identified with a vector $(\pi^{1,*},\ldots,\pi^{n,*}) \in \R^n$. Note also that the definition of a constant Nash equilibrium still requires that the optimality condition \eqref{Nash} holds for \emph{every} choice of alternative strategy, not just constant ones.
\end{remark}

Our first main finding provides conditions for existence and uniqueness of a constant Nash equilibrium and also constructs it explicitly.

\begin{theorem} \label{th:exp-eq-np}
Assume that for all $i=1,\ldots,n$ we have $\delta_i > 0$, $\theta_i \in [0,1]$, $\mu_i > 0$, $\sigma_i \ge 0$, $\nu_i \ge 0$, and $\sigma_i + \nu_i > 0$.
Define the constants 
\begin{equation}
\varphi _{n} :=\frac{1}{n}\sum_{k=1}^{n}\delta_{k}\frac{\mu_{k}\sigma_{k}}{\sigma_{k}^{2}+\nu_{k}^{2}(1-\theta_{k}/n)} \quad \quad \text{ and } \quad \quad \psi_n := \frac{1}{n}\sum_{k=1}^{n}\theta_{k}\frac{\sigma_{k}^{2}}{\sigma_{k}^{2}+\nu_{k}^{2}(1-\theta_{k}/n)}.  \label{phi-psi-n-expo}
\end{equation}
There are two cases:
\begin{enumerate}[(i)]
\item If $\psi_n < 1$, there exists a unique constant equilibrium, given by
\begin{align}
\pi^{i,*}=\delta_i\frac{\mu_i}{\sigma_i^2 + \nu_i^2(1-\theta_i/n)}+\theta_i\frac{\sigma _i}{\sigma_i^2 + \nu_i^2(1-\theta_i/n)}\frac{\varphi_n}{1-\psi_n}.  \label{eq-expo}
\end{align}
Moreover, we have the identity
\begin{align*}
\frac{1}{n}\sum_{k=1}^n\sigma_k\pi^{k,*} = \frac{\varphi_n}{1-\psi_n}.
\end{align*}
\item If $\psi_n = 1$, there is no constant equilibrium.
\end{enumerate}
\end{theorem}

An important corollary covers the special case of a single stock.

\begin{corollary}[Single stock] \label{co:exp-eq-np}
Assume that for all $i=1,\ldots,n$ we have $\mu_i=\mu > 0$, $\sigma_i=\sigma > 0$, and $\nu_i=0$. Define the constants 
\begin{equation*}
\overline{\delta} := \frac{1}{n}\sum_{k=1}^{n}\delta_k \quad \quad \text{ and } \quad \quad \overline{\theta} :=\frac{1}{n}\sum_{k=1}^{n}\theta_k.
\end{equation*}
There are two cases:
\begin{enumerate}[(i)]
\item If $\overline{\theta }<1$, there exists a unique constant equilibrium,
given by 
\begin{equation*}
\pi ^{i,*} = \left( \delta _{i} + \theta _{i}\frac{\overline{\delta }}{1-\overline{\theta }}\right)\frac{\mu }{\sigma ^{2}} .
\end{equation*}
\item If $\overline{\theta }=1$, there is no constant
equilibrium.
\end{enumerate}
\end{corollary}
\begin{proof}
Apply Theorem \ref{th:exp-eq-np}, taking note of the simplifications $\varphi_n = \overline{\delta}\mu/\sigma$ and $\psi_n = \overline{\theta}$.
\end{proof}

\begin{remark} \label{re:alternative-averaging-exp}
For a given agent $i$, it is arguably more natural to replace the average wealth $\overline{X}_T$ in the payoff functional $J_i$ defined in  \eqref{J-expo} with the average over all \emph{other} agents, \emph{not including herself}, i.e., $\overline{X}_T^{(-i)} = \frac{1}{n-1}\sum_{k \neq i}X^k_T$. 
Fortunately, there is a one-to-one mapping between the two formulations, so there is no need to solve both separately. Indeed, suppose the $i^{\text{th}}$ agent's payoff is
\[
\E\left[ -\exp \left( -\frac{1}{\delta'_i}\left( X_{T}^{i}-\theta'_i\overline{X}^{(-i)}_{T}\right) \right) \right],
\]
for some parameters $\theta'_i \in [0,1]$ and $\delta'_i>0$. By matching coefficients it is straightforward to show that
\[
\frac{1}{\delta'_i}\left( X_{T}^{i}-\theta'_i\overline{X}^{(-i)}_{T}\right)  = \frac{1}{\delta_i}\left( X_{T}^{i}-\theta_i\overline{X}_{T}\right),
\]
when $\theta_i \in [0,1]$ and $\delta_i > 0$ are defined by
\[
\delta_i = \frac{\delta'_i}{1+\frac{1}{n-1}\theta'_i} \quad\quad \text{ and } \quad\quad \theta_i = \frac{\theta'_i}{\frac{n-1}{n} + \frac{1}{n}\theta'_i}.
\]
We prefer our original formulation mainly because it results in simpler formulas for the equilibrium strategies in Theorem \ref{th:exp-eq-np} and Corollary \ref{co:exp-eq-np}. Moreover, for large $n$, this choice has negligible effect on the strategy $\pi^{i,*}$, as the differences $|\delta_i-\delta'_i|$ and $|\theta_i-\theta'_i|$ vanish.
\end{remark}

\begin{remark}
Even in the absence of competition, there are well known technical issues with exponential preferences in expected utility optimization, in that the wealth may become arbitrarily negative. This has been studied and partially addressed, essentially by carefully redefining the class of admissible controls. In particular, one can define admissible strategies such that wealth processes are supermartingales under all martingale measures with finite entropy. One must then solve the dual problem, but some technical issues may still remain; see \cite{delbaen2002exponential,schachermayer2001optimal}. On the other hand, more recent work has identified financially meaningful admissibility classes which often (but not always) contain the desired optimizer; see \cite{biagini2011admissible,biagini2012note}. 
\end{remark}

\begin{proof}[Proof of Theorem \ref{th:exp-eq-np}]
\bigskip Let $i$ be fixed.
Assume that all other agents, $k\neq i,$ follow constant
investment strategies, denoted by $\alpha_k\in \R$. Let $(X^k_t)_{t \in [0,T]}$ be the associated wealth processes, 
\begin{equation*}
X_{t}^{k}=x_{0}^{k} + \alpha_k\left( \mu _kt+\nu _kW_t^k+\sigma_kB_t\right) ,
\end{equation*}
and also define 
\begin{equation*}
Y_{t}:=\frac{1}{n}\sum_{k\neq i}X_{t}^{k}.
\end{equation*}
Then, the $i^{\text{th}}$ agent solves the optimization problem
\begin{equation}
\sup_{\pi^i \in \A}\E\left[ -\exp \left( -\frac{1}{\delta _{i}}\left(\left( 1-\frac{\theta _{i}}{n}\right) X_{T}^{i}-\theta _{i}Y_T\right)\right) \right], \label{u-expo-N}
\end{equation}
where $(X^i_t)_{t \in [0,T]}$ and $(Y_t)_{t \in [0,T]}$ have dynamics (cf. (\ref{def:X-i}))
\begin{align*}
dX^i_{t} &= \pi^i_{t}(\mu_{i}dt+\nu_{i}dW_{t}^{i}+\sigma_idB_{t}), \quad\quad\quad\quad X^i_0=x^i_0, \\
dY_{t} &= \widehat{\mu \alpha }dt + \widehat{\sigma \alpha }dB_{t}+\frac{1}{n}\sum_{k\neq i}\nu_{k}\alpha _{k}dW_{t}^{k}, \quad Y_0 = \frac{1}{n}\sum_{k\neq i}x^k_0,
\end{align*}
and where we have abbreviated
\begin{align*}
\widehat{\mu \alpha} := \frac{1}{n}\sum_{k\neq i}\mu_k\alpha_k \quad \quad \text{and} \quad \quad \widehat{\sigma\alpha} := \frac{1}{n}\sum_{k\neq i}\sigma_k\alpha_k.
\end{align*}
In the sequel, we will also use the abbreviation
\begin{align*}
\widehat{(\nu \alpha )^{2}} := \frac{1}{n}\sum_{k\neq i}\nu _{k}^{2}\alpha_{k}^{2}.
\end{align*}

The value of the supremum in \eqref{u-expo-N} is equal to $v(X^i_0,Y_0,0)$, where $v(x,y,t) $ solves the Hamilton-Jacobi-Bellman (HJB) equation 
\begin{equation}
v_{t}+\max_{\pi \in \R}\left( \frac{1}{2}(\sigma _{i}^{2}+\nu _{i}^{2})\pi^{2}v_{xx}+\pi \left( \mu _{i}v_{x}+\sigma _{i}\widehat{\sigma \alpha }v_{xy}\right) \right)   \label{HJB}
\end{equation}
\begin{equation*}
+\frac{1}{2}\left( \widehat{\sigma \alpha }^{2}+\frac{1}{n}\widehat{(\nu
\alpha )^{2}}\right) v_{yy}+\widehat{\mu \alpha }v_{y}=0,
\end{equation*}
for $\left( x,y,t\right) \in \R \times \R \times [0,T]$, with terminal condition
\begin{equation*}
v(x,y,T) = -e^{-\gamma_{i}\left(x - G(y)\right) }=-\exp \left( -\frac{1}{\delta _{i}}\left( \left( 1-\frac{\theta_{i}}{n}\right) x-\theta_{i}y\right) \right) .
\end{equation*}
Applying the first order conditions, equation (\ref{HJB}) reduces to 
\begin{equation*}
v_{t}-\frac{1}{2}\frac{(\mu _{i}v_{x}+\sigma _{i}\widehat{\sigma \alpha }v_{xy})^{2}}{(\sigma _{i}^{2}+\nu _{i}^{2})v_{xx}}+\frac{1}{2}\left( \widehat{\sigma \alpha }^{2}+\frac{1}{n}\widehat{(\nu \alpha )^{2}}\right)v_{yy}+\widehat{\mu \alpha }v_{y}=0.
\end{equation*}
Making the ansatz $v(x,y,t)=-f(t)\exp\left(-\frac{1}{\delta_i}((1-\frac{\theta_i}{n})x-\theta _{i}y)\right)$ yields, for $t \in [0,T]$,
\begin{equation*}
f'(t) - \rho f(t)=0.
\end{equation*}
with $f(T) = 1$ and  
\begin{equation}
\rho := \frac{(\mu_i+\theta_i\delta_i^{-1}\sigma_i\widehat{\sigma \alpha })^{2}}{2(\sigma_i^{2}+\nu_i^{2})} - \frac{\theta_i}{\delta_i}\widehat{\mu\alpha } - \frac{\theta_i^2}{2\delta_i^2} \left( \widehat{\sigma \alpha }^{2}+\frac{1}{n}\widehat{(\nu \alpha )^{2}}\right) . \label{def:rho-exp-np}
\end{equation}
Therefore, $f(t)=e^{-\rho (T-t)}$ and, in turn,
\begin{equation}
v\left( x,y,t\right) =-\exp\left(-\frac{1}{\delta _{i}}\left(\left(1-\frac{\theta_{i}}{n}\right)x-\theta _{i}y\right)-\rho (T-t)\right).  \label{v-exponential}
\end{equation}
The maximum in \eqref{HJB} and is achieved at
\begin{equation*}
\pi ^{i,*}(x,y,t) := -\frac{\mu _{i}v_{x}(x,y,t)+\sigma _{i}\widehat{\sigma
\alpha }v_{xy}(x,y,t)}{(\sigma _{i}^{2}+\nu _{i}^{2})v_{xx}(x,y,t)}.
\end{equation*}
Direct calculations yield that $\pi ^{i,*}$ is constant, 
\begin{equation*}
\pi ^{i,*}=\frac{\delta _{i}^{-1}(1-\theta _{i}/n)(\mu _{i}+\theta _{i}\delta _{i}^{-1}\sigma _{i}\widehat{\sigma \alpha })}{(\sigma _{i}^{2}+\nu_{i}^{2})\delta _{i}^{-2}(1-\theta _{i}/n)^{2}}=\frac{\delta _{i}\mu_{i}+\theta_{i}\sigma_{i}\widehat{\sigma \alpha }}{(\sigma _{i}^{2}+\nu_{i}^{2})(1-\theta_{i}/n)}.
\end{equation*}

We have thus constructed a smooth solution of the HJB equation and
calculated its associated feedback policy, which is constant and thus
admissible. Using the explicit form (14) and the admissibility of this
candidate control, we can establish a verification theorem following well
known arguments in stochastic optimization \cite{fleming2006controlled,pham2009continuous,touzi2012optimal}.

Alternatively, we note that the stochastic optimization problem \eqref{u-expo-N} can
be alternatively viewed as the one solved by an agent who is the ``writer''
of a liability $G(Y_T) := \frac{\theta_i}{1-\theta_i/n}Y_T$ , having
exponential preferences with risk aversion $\gamma_i := \frac{1}{\delta_i}( 1-\frac{\theta_i}{n})$. Similar problems have been studied in \cite{henderson2002valuation,musiela-zariphopoulou2004}, to which we refer the reader for more detailed arguments for the specific optimization problem at hand.

Therefore, for a candidate portfolio vector $(\alpha_{1},\ldots ,\alpha_{n})$ to
be a constant Nash equilibrium, we need $\pi^{i,*} = \alpha_{i}$, for $i=1,\ldots,n$.
Let 
\begin{equation*}
\overline{\sigma \alpha }:=\frac{1}{n}\sum_{k=1}^{n}\sigma_k\alpha_k=\widehat{\sigma \alpha }+\frac{1}{n}\sigma_i\alpha_i.
\end{equation*}
Then, we must have 
\begin{equation*}
\alpha_{i} = \pi^{i,*} =\frac{\delta _{i}\mu _{i}+\theta _{i}\sigma _{i}\overline{\sigma
\alpha }}{(\sigma _{i}^{2}+\nu _{i}^{2})(1-\theta _{i}/n)}-\frac{\theta
_{i}\sigma _{i}^{2}}{n(\sigma _{i}^{2}+\nu _{i}^{2})(1-\theta _{i}/n)}\alpha_{i},
\end{equation*}
which implies that 
\begin{align}
\alpha_{i}&=\frac{\delta _{i}\mu _{i}+\theta _{i}\sigma _{i}\overline{\sigma\alpha }}{(\sigma _{i}^{2}+\nu _{i}^{2})(1-\theta _{i}/n)}\left( 1+\frac{\theta_{i}\sigma _{i}^{2}}{n(\sigma _{i}^{2}+\nu _{i}^{2})(1-\theta _{i}/n)}\right) ^{-1} \nonumber \\
	&=\frac{\delta _{i}\mu _{i}+\theta _{i}\sigma _{i}\overline{\sigma \alpha }}{(\sigma _{i}^{2}+\nu _{i}^{2})(1-\theta _{i}/n)+\sigma _{i}^{2}\theta _{i}/n}
	=\frac{\delta _{i}\mu _{i}+\theta _{i}\sigma _{i}\overline{\sigma \alpha }}{\sigma _{i}^{2}+\nu _{i}^{2}(1-\theta _{i}/n)}.  \label{a-expo}
\end{align}
Multiplying both sides by $\sigma _{i}$ and then averaging over $i=1,\ldots
,n$, gives 
\begin{equation}
\overline{\sigma \alpha }=\varphi _{n}+\psi _{n}\overline{\sigma \alpha },
\label{eq-eqn-expo}
\end{equation}
with $\varphi _{n},\psi _{n}$ as in (\ref{phi-psi-n-expo}). For equality \eqref{a-expo} to hold, equality \eqref{eq-eqn-expo} must hold as well. There are three cases:
\begin{enumerate}[(i)]
\item If $\psi_n < 1$, then \eqref{eq-eqn-expo} yields $\overline{\sigma \alpha}=\varphi_{n}/(1-\psi _{n})$, and the equilibrium control is well defined and
given by \eqref{eq-expo}.
\item If $\psi_n=1$ and $\varphi_n > 0$, then equation \eqref{eq-eqn-expo} has no solution and thus no constant equilibria exist.
\item The remaining case is $\psi_{n}=1$ and $\varphi_{n}=0$, in which case equation \eqref{eq-eqn-expo} has infinitely many solutions. This, however, cannot occur. Indeed, because $\delta_i,\mu_i > 0$ for all $i$, we can only have $\varphi_n=0$ if $\sigma_i = 0$ for all $i$. Recalling that $\sigma_i + \nu_i > 0$ by assumption, this implies $\psi_n=0$, which is a contradiction.
\end{enumerate}
\end{proof}

\begin{remark}
One can also compute the equilibrium value function $v(x,y,t)$ of agent $i$, by explicitly computing $\rho$ defined in \eqref{def:rho-exp-np}, as the quantities $\widehat{\mu\alpha}$, $\widehat{\sigma\alpha}$, and $\widehat{\nu\alpha}$ are now known. However, we omit this tedious calculation.
\end{remark}

It remains an open problem to determine if there exist Nash equlibria that are not constant, e.g., equilibria of the feedback form $\pi^{i,*}=\pi^{i,*}(t,x_1,\ldots,x_n)$, depending on time and/or the agents' wealths. Actually, an adaptation our argument could show that the constant Nash equilibrium we derived is,  in fact, unique among the broader class of equilibria involving wealth-independent but potentially time-dependent strategies $\pi=\pi(t)$. The only delicate point is the solvability of the HJB equation \eqref{HJB}, which would force us to consider time-dependence with appropriate smoothness and growth conditions.

In addition, focusing on constant Nash equilibria can be justified as follows.
It is well known that, for log-normal models, exponential utilities lead to constant strategies. This holds not only for the plain investment problem but also, in the
absence of competition, for indifference-type investment problems.
The individual optimization problems we encountered in \eqref{u-expo-N} are directly analogous to these standard indifference-type problems.
Wealth-independence is well documented in this context, even in general semimartingale models. This is precisely the reason exponential utilities are so popular in the areas of asset price equilibrium and indifference valuation.

\subsection{The mean field game} \label{se:mfg-exp} 

In this section we study the limit as $n\rightarrow\infty$ of the $n$-agent game analyzed in the previous section.

We start with an informal argument, to build intuition and motivate the upcoming definition. For the $n$-agent game, we define for each agent $i=1,\ldots,n$ the \emph{type vector}
\begin{equation*}
\zeta_i := (x^i_0,\delta_{i},\theta_{i},\mu_{i},\nu_{i},\sigma_{i}).
\end{equation*}
These type vectors induce an empirical measure, called the \emph{type distribution}, which is the probability measure on the \emph{type space}
\begin{align}
\Z^e  := \R \times (0,\infty) \times [0,1] \times (0,\infty) \times [0,\infty) \times [0,\infty), \label{def:exp-type-space}
\end{align}
given by
\[
m_n(A) = \frac{1}{n}\sum_{i=1}^n1_A(\zeta_i), \ \text{ for Borel sets } A \subset \Z^e .
\]
We then see that for each agent $i$ the equilibrium strategy $\pi^{i,*}$ computed in Theorem \ref{th:exp-eq-np} depends only on her own type vector $\zeta_i$ and the distribution $m_n$ of all type vectors. Indeed, the constants $\varphi_n$ and $\psi_n$ (cf. \eqref{phi-psi-n-expo}) are obtained simply by integrating appropriate functions under $m_n$.

Assume now that as the number of agents becomes large, $n\rightarrow\infty$, the above empirical measure $m_n$ has a weak limit $m$, in the sense that $\int_{\Z^e } f\,dm_n \rightarrow \int_{\Z^e } f\,dm$ for every bounded continuous function $f$ on $\Z^e $. For example, this holds almost surely if the $\zeta_i$'s are i.i.d.\ samples from $m$. Let $\zeta=(\xi,\delta,\theta,\mu,\nu,\sigma)$ denote a random variable with this limiting distribution $m$. Then, we should expect the optimal strategy $\pi^{i,*}$ (cf. \eqref{eq-expo}) to converge to
\begin{align}
\lim_{n\rightarrow\infty}\pi^{i,*} =\delta_i\frac{\mu_i}{\sigma_i^2 + \nu_i^2}+\theta_i\frac{\sigma _i}{\sigma_i^2 + \nu_i^2}\frac{\varphi}{1-\psi}, \label{def:exp-limitingstrategy}
\end{align}
where
\begin{align*}
\varphi &:= \lim_{n\rightarrow\infty} \varphi_n = \E\left[\delta\frac{\mu\sigma}{\sigma^2+\nu^2}\right] \quad \quad \text{ and } \quad \quad \psi := \lim_{n\rightarrow\infty} \psi_n = \E\left[\theta\frac{\sigma^2}{\sigma^2+\nu^2}\right].
\end{align*}

The \emph{mean field game} (MFG) defined next allows us to derive the limiting strategy \eqref{def:exp-limitingstrategy} as the outcome of a self-contained equilibrium problem, which intuitively represents a game with a continuum of agents with type distribution $m$. Rather than directly modeling a continuum of agents, we follow the MFG paradigm of modeling a single representative agent, who we view as randomly selected from the population. The probability measure $m$ represents the distribution of type parameters among the continuum of agents; equivalently, the representative agent's type vector is a random variable with law $m$. Heuristically, each agent in the continuum trades in a single stock driven by two Brownian motions, one of which is unique to this agent and one of which is common to all agents. The equilibrium concept introduced in Definition \ref{def:MFE-exp} will formalize this intuition.

\subsubsection{Formulating the mean field game} \label{se:mfg-exp-formulation}
To formulate the MFG, we now assume that the probability space $(\Omega,\F,\PP)$ supports yet another independent (one-dimensional) Brownian motion, $W$, as well as a random variable
\[
\zeta = (\xi,\delta ,\theta ,\mu ,\nu ,\sigma),
\]
independent of $W$ and $B$, and with values in the space $\Z^e $ defined in \eqref{def:exp-type-space}. This random variable $\zeta$ is called the \emph{type vector}, and its distribution is called the \emph{type distribution}.

Let $\FF^{\mathrm{MF}} = (\F^{\mathrm{MF}}_t)_{t \in [0,T]}$ denote the smallest filtration satisfying the usual assumptions for which $\zeta $ is $\F^{\mathrm{MF}}_0$-measurable and both $W$ and $B$ are adapted. Let also $\FF^B=(\F^B_t)_{t \in [0,T]}$ denote the natural filtration generated by the Brownian motion $B$.\footnote{One might as well formulate the MFG on a different probability space from the $n$-agent game, but we prefer to avoid introducing additional notation.}

The \emph{representative agent's} wealth process solves
\begin{align}
dX_t = \pi_t(\mu dt + \nu dW_t + \sigma dB_t), \quad X_0 = \xi, \label{def:X-MFG}
\end{align}
where the portfolio strategy must belong to the admissible set $\A_{\mathrm{MF}}$ of self-financing $\FF^{\mathrm{MF}}$-progressively measurable real-valued processes $(\pi_t)_{t \in [0,T]}$ satisfying $\E\int_0^T|\pi_t|^2dt < \infty$.
The random variable $\xi$ is the initial wealth of the representative agent, whereas $(\mu,\nu,\sigma)$ are the market parameters. In the sequel, the parameters $\delta$ and $\theta$ will affect the risk preferences of the representative agent.

In this mean field setup, the \emph{single stock case} refers to the case where $(\mu,\nu,\sigma)$ is nonrandom, with $\nu=0$, $\mu > 0$, and $\sigma > 0$. In the context of the limiting argument above, this corresponds to the $n$-agent game in which $\mu_i=\mu$, $\nu_i=\nu=0$, and $\sigma_i=\sigma$ for all $i$. Note that each agent among the continuum may still have different preference parameters, captured by the fact that $\delta$ and $\theta$ are random.

\begin{remark}
There are two distinct sources of randomness in this model. One source comes from the Brownian motions $W$ and $B$ which drive the stock price processes over time. 
The second source is static and comes from the random variable $\zeta$, which describes the distribution of type vectors (which includes initial wealth, individual preference parameters, and market parameters) among a large (in fact, continuous) population. One
can then think of a continuum of agents, each of whom is assigned an i.i.d.\ type
vector at time zero, and the agents interact \emph{after} these assignments
are made.
\end{remark}

To see how to formulate the representative agent's optimization problem, let us first recall how the Nash equilibrium in the $n$-agent game was constructed. We first solved the optimization problem \eqref{HJB} faced by each single agent $i$, in which the strategies of the other agents $k \neq i$ were treated as fixed. However, instead of fixing the strategies of the other agents, we could have just fixed the mean terminal wealth $\frac{1}{n}\sum_{k\neq i}X^k_T$, as this is effectively the only source of interaction between the agents. This was precisely the idea behind the proof of Theorem \ref{th:exp-eq-np}, and this guides the upcoming formulation of the MFG.

To this end, suppose that $\overline{X}$ is a given random variable, representing the average wealth of the continuum of agents.
The representative agent has no influence on $\overline{X}$, as but one agent amid a continuum.
The objective of the representative agent is thus to maximize the expected payoff
\begin{align}
\sup_{\pi \in \A_{\mathrm{MF}}}\E\left[-\exp\left(-\frac{1}{\delta}\left(X_T - \theta\overline{X}\right)\right)\right], \label{def:exp-mfg-optimization}
\end{align}
where $(X_t)_{t \in [0,T]}$ is given by \eqref{def:X-MFG}. We are now ready to introduce the main definition of this section.

\begin{definition} \label{def:MFE-exp}
Let $\pi^* \in \A_{\mathrm{MF}}$ be an admissible strategy, and consider the $\F^B_T$-measurable random variable $\overline{X} := \E[X^*_T \,|\, \F^B_T]$, where $(X^*_t)_{t \in [0,T]}$ is the wealth process in \eqref{def:X-MFG} corresponding to the strategy $\pi^*$. We say that $\pi^*$ is a mean field equilibrium (MFE) if $\pi^*$ is optimal for the optimization problem \eqref{def:exp-mfg-optimization} corresponding to this choice of $\overline{X}$.

A constant MFE is an $\F^{\mathrm{MF}}_0$-measurable random variable $\pi^*$ such that, if $\pi_t := \pi^*$ for all $t \in [0,T]$, then $(\pi_t)_{t \in [0,T]}$ is a MFE. 
\end{definition}

Typically, a MFE is computed as a fixed point. One starts with a generic $\F^B_T$-measurable random variable $\overline{X}$, solves \eqref{def:exp-mfg-optimization} for an optimal $\pi^*$, and then computes $\E[X^*_T \,|\, \F^B_T]$. If we have a fixed point in the sense that the \emph{consistency condition}, $\E[X^*_T \,|\, \F^B_T] = \overline{X}$, holds, then $\pi^*$ is a MFE.
Intuitively, every agent in the continuum faces an independent noise $W$, an independent type vector $\zeta$, and the same common noise $B$. Therefore, conditionally on $B$, all agents face i.i.d.\ copies of the same optimization problem. Heuristically, the law of large numbers suggests that the average terminal wealth of the whole population should be $\E[X^*_T \,|\, \F^B_T]$. This consistency condition illustrates the distinct roles played by the two Brownian motions $W$ and $B$ faced by the representative agent.

Perhaps more intuitively clear is the case where $\sigma=0$ a.s., so there is no common noise term. In this case, the consistency condition could be replaced with $\overline{X} = \E[X^*_T]$, owing to the fact that each agent in the continuum faces an i.i.d.\ copy of the same optimization problem.

We refer the reader to \cite{carmona-delarue-lacker} for a detailed discussion of mean field games with common noise, and to \cite{cardaliaguet2015master,lacker2014general} for general results on limits of $n$-agent games. Alternatively, the so-called ``exact law of large numbers'' provides another way to formalize this idea of averaging over a continuum of (conditionally) independent agents \cite{sun2006exact}.

\subsubsection{An alternative formulation of the mean field game} \label{se:alternative-mfg}
It is worth emphasizing that the optimization problem \eqref{def:exp-mfg-optimization} treats the type vector $\zeta$ as a genuine source of randomness, in addition to the stochasticity coming from the Brownian motions. However, an alternative interpretation is given below which will also help in solving the MFG.

As our starting point, note that for a fixed $\F^B_T$-measurable random variable $\overline{X}$ we have
\begin{equation}
\sup_{\pi \in \A_{\mathrm{MF}}}\E\left[ -e^{-\frac{1}{\delta }\left( X_T-\theta \overline{X}\right) }\right] = \E\left[ u(\zeta)\right], \label{def:exp-mfg-optimization-rewrite}
\end{equation}
where $u(\cdot)$ is a value function defined for (deterministic) elements $\zeta_0=(x_0,\delta_0,\theta_0,\mu_0,\nu_0,\sigma_0)$ of the type space $\Z^e$ by
\begin{align}
u(\zeta_0) &:= \sup_{\pi}\E\left[-\exp \left( -\frac{1}{\delta_0}\left( \widetilde{X}_{T}^{\zeta_0,\pi }-\theta_0 \overline{X}\right) \right)\right], \label{def:reformulation}
\end{align}
with
\[
d\widetilde{X}^{\zeta_0,\pi}_t = \pi_t\left(\mu_0dt + \nu_0dW_t + \sigma_0dB_t\right), \  \widetilde{X}^{\zeta_0,\pi}_0 = x_0,
\]
and where the supremum is over square-integrable  processes which are progressively measurable with respect to the filtration generated by the Brownian motions $W$ and $B$ (noting that the random variable $\zeta$ is absent from this filtration).

For a deterministic type vector $\zeta_0 \in \Z^e$, the quantity $u(\zeta_0)$ can be interpreted as the value of the optimization problem \eqref{def:reformulation} \emph{faced by an agent of type $\zeta_0$}. On the other hand, the original optimization problem on the left-hand side of \eqref{def:exp-mfg-optimization-rewrite} gives the optimal expected value faced by an agent \emph{before the random assignment of types} at time $0$.

This new interpretation will be used somewhat implicitly to compute a MFE in the proof of Theorem \ref{th:exp-eq} below, in the following manner. We may write $u(\zeta_0)=v_{\zeta_0}(x_0,0)$ as the time-zero value of the solution of a HJB equation, $v_{\zeta_0}(x,t)$, for $(x,t) \in \R \times [0,T]$.\footnote{There is some redundancy in this notation, as $x_0$ is already part of the vector $\zeta_0$.} The optimal value on the left-hand side of \eqref{def:exp-mfg-optimization-rewrite} is then the expectation of these time-zero values, $\E[v_{\zeta}(\xi,0)]$, when the \emph{random} type vector $\zeta$ is used.

Similarly, the optimal strategy $\pi^{\zeta_0,*}$ in \eqref{def:reformulation} depends on the fixed value of $\zeta_0$. The optimal strategy for the left-hand side of \eqref{def:exp-mfg-optimization-rewrite} is then obtained by plugging in the random type vector, yielding $\pi^{\zeta,*} \in \A_{\mathrm{MF}}$. 
This justifies the interpretation of the strategy $\pi^{\zeta_0,*}$ as the strategy chosen by an agent with type vector $\zeta_0$.

\subsubsection{Solving the mean field game}
Next, we present the second main result, in which we construct a constant MFE and also provide conditions for its uniqueness. The result also confirms that the MFG formulation is indeed appropriate, as the MFE we obtain agrees with the limit of the $n$-agent equilibrium strategies in the sense of \eqref{def:exp-limitingstrategy}.

\begin{theorem} \label{th:exp-eq}
Assume that, a.s., $\delta > 0$, $\theta \in [0,1]$, $\mu > 0$, $\sigma \ge 0$, $\nu \ge 0$, and $\sigma+\nu > 0$. Define the constants 
\begin{equation*}
\varphi := \E\left[ \delta\frac{\mu \sigma }{\sigma^{2} + \nu^{2}}\right] \quad \quad \text{ and } \quad \quad  \psi := \E\left[ \theta \frac{\sigma ^{2}}{\sigma^{2} + \nu^{2}}\right],
\end{equation*}
where we assume that both expectations exist and are finite.

There are two cases:
\begin{enumerate}
\item If $\psi < 1$, there exists a unique constant MFE, given by
\begin{equation}
\pi^* = \delta \frac{\mu }{\sigma^{2} + \nu^{2}} + \theta \frac{\sigma}{\sigma^{2} + \nu^{2}}\frac{\varphi}{1-\psi}.  \label{mf-expo-optimal-control}
\end{equation}
Moreover, we have the identity
\begin{align*}
\E[\sigma\pi^*] = \frac{\varphi}{1-\psi}.
\end{align*}
\item If $\psi =1$, there is no constant MFE.
\end{enumerate}
\end{theorem}

Next, we highlight the single stock case, noting that the form of the solution is essentially the same as in the $n$-agent game, presented in Corollary \ref{co:exp-eq-np}.

\begin{corollary}[Single stock] \label{co:exp-eq-single}
Suppose $(\mu,\nu,\sigma)$ are deterministic, with $\nu = 0$ and $\mu,\sigma > 0$. Define the constants
\begin{equation*}
\overline{\delta} := \E\left[ \delta \right] \quad \quad \text{ and } \quad \quad \bar{\theta} := \E[\theta ].
\end{equation*}
There are two cases:
\begin{enumerate}
\item If $\overline{\theta} < 1$, there exists a unique constant MFE, given by 
\begin{equation*}
\pi^* = \left( \delta +\theta \frac{\overline{\delta}}{1-\overline{\theta}}\right)\frac{\mu }{\sigma^2}.
\end{equation*}
\item If $\overline{\theta} = 1$, there is no constant MFE.
\end{enumerate}
\end{corollary}

\begin{proof}[Proof of Theorem \ref{th:exp-eq}]
The first step in constructing a constant MFE is to solve the
stochastic optimization problem in (\ref{def:exp-mfg-optimization}), for a given choice of $\overline{X}$.
First, observe that it suffices to restrict our attention to random variables $\overline{X}$ of
the form $\overline{X} = \E[X^\alpha_T|\F^B_T]$, where $X^\alpha$ solves \eqref{def:X-MFG} for some admissible strategy $\alpha \in \A_{\mathrm{MF}}$. However, because we are searching only for constant MFE, we may fix a constant strategy, i.e., an $\F^{\mathrm{MF}}_{0}$-measurable random variable $\alpha$ with $\E[\alpha^2] < \infty$.

It is convenient to define, for $t \in [0,T]$, 
\begin{equation*}
\overline{X}_{t} := \E[X_{t}^\alpha|\F^B_T].
\end{equation*}
Noting that $\overline{X}_T = \overline{X}$, the key idea is then to identify the dynamics of the process $(\overline{X}_t)_{t \in [0,T]}$ and incorporate it into the state process of the control
problem (\ref{def:exp-mfg-optimization}). Because $(\xi ,\mu ,\sigma,\nu ,\alpha )$, $W$, and $B$ are independent, we must have 
\begin{equation*}
\overline{X}_{t} = \overline{\xi}+\overline{\mu \alpha }t+\overline{\sigma \alpha }B_{t},
\end{equation*}
where we use the notation $\overline{M} = \E[M]$ for an integrable random variable $M$.

In turn, for $\pi \in \A_{\mathrm{MF}}$, we define, for $t \in [0,T]$,
\begin{equation*}
Z_{t}^{\pi }:=X_{t}^{\pi }-\theta \overline{X}_{t},
\end{equation*}
with $(X_{t}^{\pi })_{t \in [0,T]}$ solving \eqref{def:X-MFG}. Then, 
\begin{equation*}
dZ_{t}^{\pi }=(\mu \pi _{t}-\theta \overline{\mu \alpha })dt+\nu \pi
_{t}dW_{t}+(\sigma \pi _{t}-\theta \overline{\sigma \alpha })dB_{t},
\end{equation*}
with $Z^\pi_0 = \xi - \theta\overline{\xi}$.

We have thus absorbed $\overline{X}$ as part of the controlled state process. As a result, instead of solving the original control problem (\ref{def:exp-mfg-optimization}) we can equivalently solve the Merton-type problem,
\begin{equation}
\sup_{\pi \in \A_{\mathrm{MF}}}\E\left[ -\exp\left(-\frac{1}{\delta }Z_{T}^{\pi }\right)\right]. \label{pf:exp-mfg-optimization-reduced}
\end{equation}
As in the discussion in Section \ref{se:alternative-mfg}, the above supremum equals $\E[v(\xi - \theta\overline{\xi},0)]$, where $v(x,t)$ is the unique (smooth, strictly concave and strictly increasing in $x$) solution of the HJB equation 
\begin{equation}
v_{t}+\max_{\pi }\left( \frac{1}{2}\left( \nu ^{2}\pi ^{2}+(\sigma \pi-\theta \overline{\sigma \alpha })^{2}\right) v_{xx}+(\mu \pi -\theta \overline{\mu \alpha })v_{x}\right) =0,  \label{HJB-mf-expo}
\end{equation}
with terminal condition $v(x,T)=-e^{-x/\delta}$. 
We stress that this HJB equation is random, in the sense that it depends on
the $\F^{\mathrm{MF}}_0$-measurable type parameters $(\delta,\theta,\mu,\nu,\sigma)$. 

Equation \eqref{HJB-mf-expo} simplifies to
\begin{equation*}
v_{t}-\frac{1}{2}\frac{\left( \mu v_{x}-\theta \sigma \overline{\sigma\alpha }v_{xx}\right) ^{2}}{(\sigma^{2}+\nu^{2})v_{xx}}-\theta \overline{\mu \alpha }v_{x} + \frac12 (\theta\overline{\sigma \alpha})^2 v_{xx} =0.
\end{equation*}
Making the ansatz $v(x,t)=-e^{-x/\delta }f(t)$, the above reduces to 
\begin{equation*}
f^{\prime }(t) - \rho f(t)=0,
\end{equation*}
with $f(T)=1$, and with $\rho $ given by the $\F^{\mathrm{MF}}_0$-measurable random variable 
\begin{equation}
\rho := \frac{\left( \mu+ \frac{\theta}{\delta} \sigma \overline{\sigma \alpha }\right) ^{2}}{2(\sigma^{2}+\nu^{2})} - \frac{\theta}{\delta } \overline{\mu \alpha } - \frac12\left(\frac{\theta}{\delta}\overline{\sigma \alpha}\right)^2. \label{def:rho-exp-mfg}
\end{equation}
Thus, $f(t)=e^{-\rho (T-t)}$, and $v(x,t)=-e^{-x/\delta }f(t)$.
Furthermore, the optimal feedback control achieving the maximum in \eqref{HJB-mf-expo} is given by
\begin{align}
\pi^*(x,t) &= -\frac{\mu v_{x}\left( x,t\right) -\theta\sigma \overline{\sigma \alpha }v_{xx}(x,t) }{(\sigma^{2}+\nu^{2})v_{xx}(x,t)}   = \delta\frac{\mu}{\sigma^{2}+\nu^{2}}+\theta\frac{ \sigma }{\sigma^{2}+\nu^{2}}\overline{\sigma \alpha },  \label{pf:exp-optimal-control}
\end{align}
In fact, $\pi^*=\pi^*(x,t)$ is $\F^{\mathrm{MF}}_0$-measurable and does not depend on $(x,t)$. The optimality of $\pi^*$ for the problem \eqref{pf:exp-mfg-optimization-reduced} follows.

Recalling Definition \ref{def:MFE-exp}, we see that for the candidate control $\alpha$ to be a constant MFE, we need $\alpha = \pi^*$. In light of \eqref{pf:exp-optimal-control}, $\pi^*$ is a constant MFE if it solves the equation
\begin{align}
\pi^* &= \delta\frac{\mu}{\sigma^{2}+\nu^{2}}+\theta\frac{ \sigma }{\sigma^{2}+\nu^{2}}\overline{\sigma \pi^* }.  \label{pf:exp-optimal-control2}
\end{align}
Multiply both sides by $\sigma$ and average to find that $\overline{\sigma \pi^*}$ must satisfy
\begin{equation}
\overline{\sigma \pi^* }= \E\left[ \delta\frac{\mu \sigma  }{\sigma ^{2}+\nu^{2}}\right] + \E\left[ \theta\frac{ \sigma ^{2}}{\sigma ^{2}+\nu^{2}}\right] \overline{\sigma \pi^* }=\varphi +\psi \overline{\sigma \pi^* }.  \label{eqn-aux}
\end{equation}
We then have the following cases:
\begin{enumerate} [(i)]
\item If $\psi <1$, the above yields $\overline{\sigma \pi^* }=\varphi/(1-\psi )$, and using equation \eqref{pf:exp-optimal-control2} we prove part (1).
\item If $\psi =1$ but $\varphi \neq 0$, then equation (\ref{eqn-aux}) has no
solution, and, as a result, there can be no constant MFE.
\item The remaining case is $\psi=1$ and $\varphi =0$. However, this cannot happen. Indeed, if this were the case, then $\varphi = 0$ and the restrictions on the parameters would imply $\sigma=0 < \nu$ a.s., which would in turn yield $\psi=0$, a contradiction. This completes the proof of part (2).
\end{enumerate}
\end{proof}

\begin{remark}
Note that the proof above yields a tractable formula for the equilibrium value function of the representative agent. Since the controlled process $(Z^\pi_t)_{t \in [0,T]}$ starts from $Z^\pi_0=\xi - \theta\overline{\xi}$, the time-zero value to the representative agent (also called $u(\zeta)$ in Section \ref{se:alternative-mfg}) is given by
\begin{align}
v(\xi - \theta\overline{\xi},0)= -\exp\left(-\frac{1}{\delta}(\xi - \theta\overline{\xi}) - \rho T\right). \label{def:master-eq}
\end{align}
It is now straightforward to explicitly compute $\rho$ in \eqref{def:rho-exp-mfg}, using the values of $\overline{\mu\alpha}$ and $\overline{\sigma\alpha}$. Indeed,
\begin{align}
\rho &= \frac{1}{2(\sigma^2+\nu^2)}\left(\mu + \frac{\theta}{\delta}\frac{\varphi}{1-\psi}\sigma\right)^2 - \frac{\theta}{\delta}\left(\widetilde{\psi} + \frac{\widetilde{\varphi}\varphi}{1-\psi}\right), \label{def:rho-exp-mfg-result}
\end{align}
where the constants $\widetilde{\varphi}$ and $\widetilde{\psi}$ are defined by
\begin{align*}
\widetilde{\psi} = \E\left[\delta\frac{\mu^2}{\sigma^2+\nu^2}\right] \quad\quad \text{ and } \quad\quad \widetilde{\varphi} = \E\left[\theta\frac{\mu\sigma}{\sigma^2+\nu^2}\right].
\end{align*}
Notably, in the single stock case, this simplifies further to
\begin{align*}
\rho &= \left(1 + \left(\frac{\overline{\delta}\theta}{\delta(1-\overline{\theta})}\right)^2\right)\frac{\mu^2}{2\sigma^2}.
\end{align*}
\end{remark}

\begin{remark} \label{re:master-eq}
Equation \eqref{def:master-eq} essentially provides the solution of the so-called \emph{master equation}; see \cite{bensoussan2015master} or \cite{carmona2014master} for an introduction to the master equation in MFG theory. Indeed, the master equation is a PDE providing a function $U=U(x,m,t)$, where $(x,t) \in \R \times [0,T]$ and $m$ is a probability measure on $\R$. The value $U(x,m,t)$ is naturally interpreted as the value at time $t$ for a representative agent starting with wealth $X_t=x$, when the distribution of other agents' wealths is $m$. In our case, this value is nothing but
\[
U(x,m,t) = v(x-\theta \bar{m},t) = -\exp\left(-\frac{1}{\delta}(x - \theta \bar{m}) - \rho(T-t)\right),
\]
where $\bar{m}$ is the mean of the measure $m$.
We do not attempt to make this any more precise, as the concept of a master equation has not yet been settled for models with different types of agents.
\end{remark}

\subsection{Discussion of the equilibrium} \label{se:discussion-exp}
We focus most of the discussion on the mean field equilibria of Theorem \ref{th:exp-eq} and Corollary \ref{co:exp-eq-single}, as the $n$-agent equilibria of Theorem \ref{th:exp-eq-np} and Corollary \ref{co:exp-eq-np} have essentially the same structure.

Recall first that the MFE $\pi^*$ is $\F^{\mathrm{MF}}_0$-measurable or, equivalently, $\zeta$-measurable where $\zeta =(\xi,\delta ,\theta ,\mu ,\nu ,\sigma )$ is the type vector. The randomness of $\zeta$ captures the distribution of type vectors among the population, while a single realization of $\zeta$ can be interpreted as the type vector of a single representative agent. Hence, we interpret the investment strategy $\pi^{*}$ as the equilibrium strategy \emph{adopted by those agents with type vector $\zeta $}.

The equilibrium portfolio $\pi^*$ consists of two components. The first, $\delta\mu/(\sigma^2+\nu^2)$, is the classical Merton portfolio in the absence of relative performance concerns.
The second component
is always nonnegative, vanishing only in the absence of competition, i.e., when $\theta=0$. It increases linearly with the competition weight $\theta$, so we find that competition always increases the allocation in the risky asset. 

The representative agent's strategy $\pi^*$ is influenced by the other agents only through the  quantity $\varphi /(1-\psi) = \E[\sigma\pi^*]$. This quantity can be viewed as the volatility of aggregate wealth. Indeed, let $X^*$ denote the wealth process corresponding to $\pi^*$ (i.e., the solution of \eqref{def:X-MFG}). The average wealth of the population at time $t \in [0,T]$ is $Y_t: = \E[X^*_t | \F^B_T]$. A straightforward computation using the independence of $\zeta$, $W$, and $B$ yields
\[
Y_t = \E[\xi] + \E[\mu\pi^*]t + \E[\sigma\pi^*]B_t.
\]

Alternatively, we may interpret the ratio $\varphi/(1-\psi)$ in terms of the type distribution.
Define $R = \sigma^2/(\sigma^2+\nu^2)$, which is the fraction of the representative agent's stock's variance driven by the common noise $B$. Then $\varphi = \E[R\delta\mu/\sigma]$ is computed by multiplying each agent's Sharpe ratio by her risk tolerance parameter and the weight $R$, then averaging over all agents.
Similarly, $\psi= \E[R\theta]$ is the average competition parameter, weighted by $R$.
Several important factors will lead to an increase in $\varphi/(1-\psi)$, and thus an increase in the investment $\pi^*$ in the risky asset. Namely, $\pi^*$ increases as other agents become more risk tolerant (higher $\delta$ on average), as other agents become more competitive (higher $\theta$ on average), or as the quality of the other stocks increases, as measured by their Sharpe ratio (higher $\mu/\sigma$ on average).

Some of the effects of competition are more transparent in the single stock case of Corollary \ref{co:exp-eq-single}. The resulting MFE $\pi^*$ clearly resembles the Merton portfolio but with \emph{effective risk tolerance parameter} 
\[
\delta_{\mathrm{eff}}:=\delta +\theta \frac{\overline{\delta }}{1-\overline{\theta }}. 
\]
We always have $\delta_{\mathrm{eff}}>\delta $ if $\theta >0$, and the
difference $\delta_{\mathrm{eff}}-\delta $ increases with $\theta$, with $\overline{\delta}$, and with $\overline{\theta}$. That is, the representative agent invests more in the risky asset if she is more competitive, if other agents tend to be more risk tolerant, or if other agents tend to be more competitive. In the latter cases, when $\overline{\delta}$ and $\overline{\theta}$ increase, we can interpret the increase in $\pi^*$ as an effort, on the part of the representative agent, to ``keep up'' with a population more willing to take risk.
At the extreme ends, as both $\theta$ and $\bar{\theta}$ approach $1$, $\pi^*$ blows up very quickly; that is, a highly competitive agent in a population of highly competitive agents invests significantly in the risky asset. This is illustrated in Figure \ref{fig:CARA1} below.

\begin{figure}[h]
\centering
\includegraphics[scale=0.7]{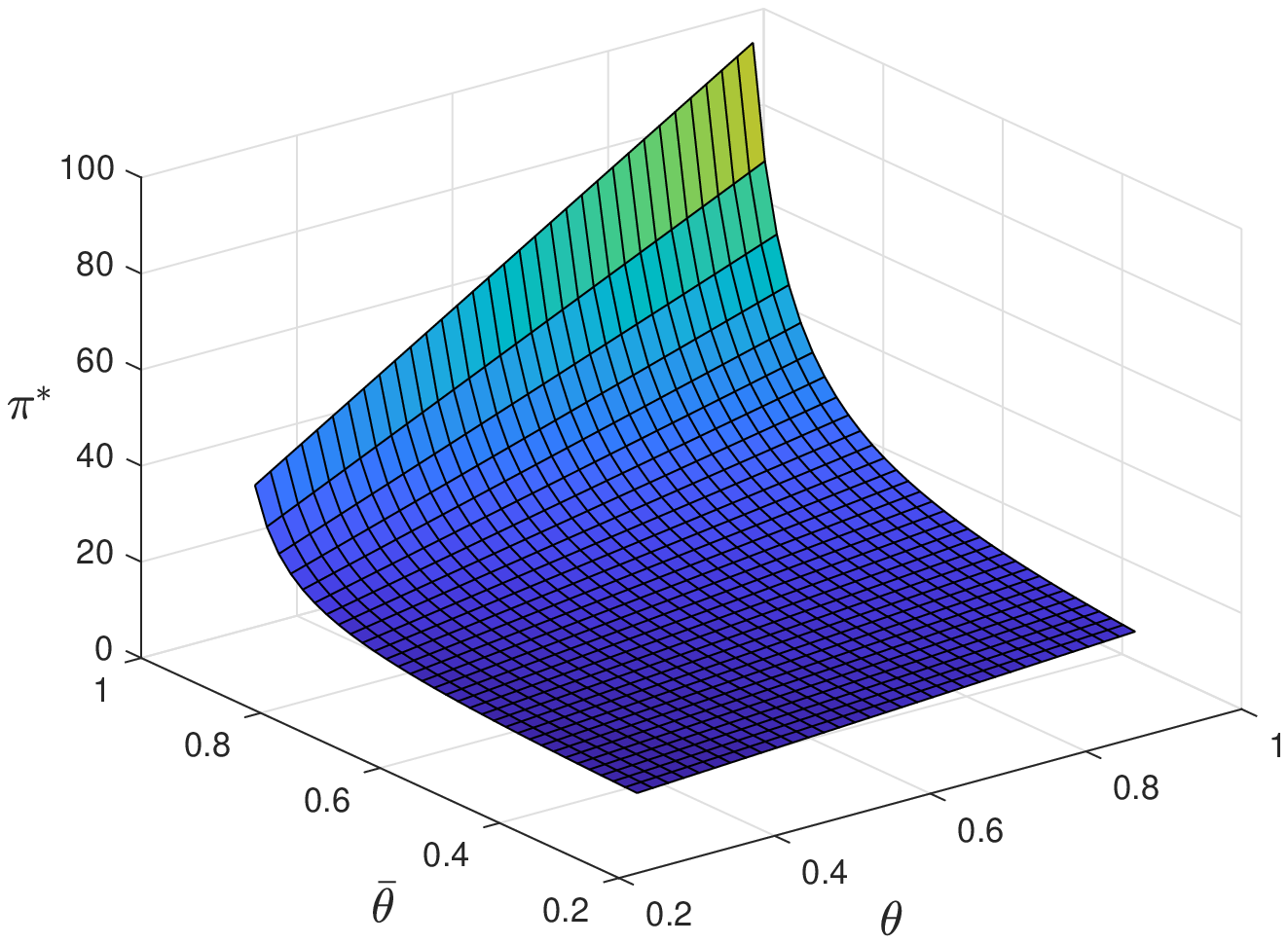}
\caption{Single stock case (Corollary \ref{co:exp-eq-single}): $\pi^*$ versus $\theta$ and $\bar{\theta}$, with $\delta =5$, $\bar{\delta}=6$, and $\mu=\sigma=1$.}
\label{fig:CARA1}
\end{figure}

A few other special cases are worth discussing. If $\sigma =0$ a.s., there is no common noise. In this case, $\varphi =\psi =0,$ and in turn the MFE is equal to the Merton portfolio. All agents act independently uncompetitively, not taking into account the performance of their competitors.

On the other hand, if $\nu =0$ a.s., there is no independent noise, and we have the simplifications $\psi=\E[\theta]$ and $\varphi=\E[\delta\mu/\sigma]$. If $\E[\theta] <1$, then 
\[
\pi^{*} = \delta \frac{\mu }{\sigma ^{2}}+\frac{\theta }{\sigma \left( 1-\E\left[ \theta \right] \right)} \E\left[ \delta \frac{\mu }{\sigma}\right] . 
\]
If $\nu =0$ a.s. and also $\E[\theta] =1$, then $\theta =1$ a.s. and $\psi=1$. In this case, every agent is concerned exclusively with relative and not absolute performance, and there is no equilibrium.

Another degenerate case is when all agents have the same type vector, i.e., when $\zeta $ is deterministic. Then, the MFE is common for all agents and (assuming $\theta <1$) reduces to
\[
\pi^{*} = \frac{\delta \mu }{\left( 1-\theta \right) \sigma ^{2}+\nu^{2}}.
\]

Lastly, we comment on the effect of population size on the equilibrium in the $n$-agent game given in Theorem \ref{th:exp-eq-np}. The only real difference compared to the mean field setting is the rescaling of $\nu_k^2$ by $(1-\theta_k/n)$ wherever it appears in Theorem \ref{th:exp-eq-np}, with no rescaling present in the single stock case of Corollary \ref{co:exp-eq-np}.
It is not yet clear how to properly interpret this rescaling, and it is worth noting that the change of variables discussed in Remark \ref{re:alternative-averaging-exp} does not significantly change the situation. Interpreting the average wealth $\frac{1}{n-1}\sum_{j\neq i}X^j_T$ as a liability term in an indifference valuation problem, as we have mentioned before, seems promising, but we do not pursue this any further here.

\section{CRRA risk preferences} \label{se:power}

In this section we focus on power and logarithmic (CRRA) utilities. Given
the homogeneity properties of the power risk preferences, we choose to
measure relative performance using a multiplicative and not additive factor.
Such cases were analyzed for a two-agent setting in \cite{basak2015competition} and
more recently in \cite{anthropelos-geng-zariphopoulou} under forward relative performance criteria.

\subsection{The $n$-agent game}

\label{se:np-power}
We consider an $n$-agent game analogous to that of Section \ref{se:np-exp}, but where each agent has a CRRA utility. We work on the same filtered probability space of Section \ref{se:np-exp}, and we assume that the $n$ stocks have the same dynamics as in \eqref{stock}.

The $n$ agents trade in a common investment horizon. As is common in power utility models, the strategy $\pi^{i}_t$ is taken to be the \emph{fraction} (as opposed to the amount) of wealth that agent $i$ invests in the stock $S^{i}$ at time $t$. Her discounted wealth is then given by
\begin{equation}
dX^{i}_{t} = \pi^{i}_{t}X^{i}_{t} \left( \mu_{i} dt + \nu_{i} dW^i_{t} + \sigma_{i} dB_{t}\right), \label{def:X-i-power}
\end{equation}
with initial endowment $X^{i}_{0}=x^i_0$. The class of admissible strategies is as before the set $\A$ of self-financing $\FF$-progressively measurable processes $(\pi_t)_{t \in [0,T]}$ satisfying $\E\int_0^T|\pi_t|^2dt < \infty$.

The $i^{\text{th}}$ agent's utility is a function $U_i : \R_+^2 \rightarrow \R$ of both her
individual wealth, $x$, and the geometric average wealth of all agents, $m$. Specifically,
\[
U_i(x,m) := U(x m^{-\theta_i};\delta_i),
\]
where $U(x;\delta)$ is defined for $x > 0$ and $\delta > 0$ by
\begin{align*}
U(x;\delta) := \begin{cases}
\left(1-\frac{1}{\delta}\right)^{-1}x^{1-\frac{1}{\delta}}, &\text{for } \delta \neq 1, \\
\log x &\text{for } \delta = 1.
\end{cases}
\end{align*}
The constant parameters $\delta_i > 0$ and $\theta_i \in [0,1]$ are the personal relative risk tolerance and competition weight parameters, respectively.\footnote{For CRRA utilities, it is more common to use the \emph{relative risk aversion} parameter $\gamma_{i}=1/\delta_{i}$, but our choice of parametrization ensures that
the \emph{relative risk tolerance} is precisely 
\begin{equation*}
\delta_i = -\frac{U_{x}(x;\delta_i)}{xU_{xx}(x;\delta_i)}.
\end{equation*}}
If agents $i=1,\ldots,n$ choose admissible strategies $\pi^1,\ldots,\pi^n$, the payoff for agent $i$ is given by 
\begin{equation}
J_i(\pi^1,\ldots,\pi^n) := \E\left[ U\left(X_{T}^{i}\overline{X}_T^{-\theta_i};\delta_i\right)\right] ,\quad \text{where} \quad \overline{X}_{T}=\left(\prod_{k=1}^{n}X_{T}^{k}\right) ^{1/n}. \label{def:Jpower}
\end{equation}
Notice that here, unlike in the exponential utility model, agents measure relative wealth using the \emph{geometric mean}, rather than the arithmetic mean. Working with the geometric mean instead of the arithmetic mean renders the problem tractable, as it allows us to exploit the homogeneity of the utility function.

The above expected utility may be rewritten more illustratively as
\[
J_i(\pi^1,\ldots,\pi^n) = \E\left[ U\left((X_T^i)^{1-\theta_i}(R^i_T)^{\theta_i};\delta_i\right)\right],
\]
where $R_{T}^{i}=X_{T}^{i}/\overline{X}_{T}$ is the \emph{relative return}
for agent $i$. This clarifies the role of the competition weight $\theta_i$ as governing the trade-off between absolute and relative wealth to agent $i$, as in the exponential utility model.
As before, an agent with a higher value of $\theta_i$ is more concerned with relative wealth than with absolute wealth.

The notion of (Nash) equilibrium is defined exactly as in Definition \ref{def:Nash}, but with the new objective function defined in \eqref{def:Jpower} above.
We find a unique \textit{constant} equilibrium in the following theorem, which we subsequently specialize to the single stock case.

\begin{theorem} \label{th:power-eq-np}
Assume that for all $i=1,\ldots,n$ we have $x^i_0 > 0$, $\delta_i > 0$, $\theta_i \in [0,1]$, $\mu_i > 0$, $\sigma_i\ge 0$, $\nu_i\ge 0$, and $\sigma_i+\nu_i > 0$.
Define the constants
\begin{equation}
\varphi_n := \frac{1}{n}\sum_{k=1}^{n}\delta_{k}\frac{\mu_{k}\sigma_{k}}{\sigma_{k}^{2}+\nu_{k}^{2}(1+(\delta_k - 1)\theta_{k}/n)}
\label{phi-n-power}
\end{equation}
and
\begin{equation}
\psi_n := \frac{1}{n}\sum_{k=1}^{n}\theta_{k}(\delta_{k}-1)\frac{\sigma_{k}^{2}}{\sigma_{k}^{2}+\nu_{k}^{2}(1+(\delta_k - 1)\theta_{k}/n)}.  \label{psi-n-power}
\end{equation}
There exists a unique constant equilibrium, given by 
\begin{equation}
\pi^{i,*} =\delta_i\frac{\mu _{i}}{\sigma _{i}^{2}+\nu _{i}^{2}(1+(\delta_i - 1)\theta _{i}/n)} - \theta _{i}(\delta_i-1) \frac{\sigma _{i}}{\sigma _{i}^{2}+\nu _{i}^{2}(1+(\delta_i - 1)\theta _{i}/n)}\,\frac{\varphi_n}{1+\psi_n}.  \label{eq-power}
\end{equation}
Moreover, we have the identity
\begin{align*}
\frac{1}{n}\sum_{k=1}^n\sigma_k\pi^{k,*} = \frac{\varphi_n}{1+\psi_n}.
\end{align*}
\end{theorem}

\begin{corollary}[Single stock] \label{co:power-eq-np}
Assume that for all $i=1,\ldots,n$ we have $\mu_i=\mu > 0$, $\sigma_i=\sigma > 0$, and $\nu_i=0$. Define the constants 
\begin{equation*}
\overline{\delta } := \frac{1}{n}\sum_{k=1}^{n}\delta_{k} \quad \quad \text{ and } \quad \quad  \overline{\theta (\delta -1)} := \frac{1}{n}\sum_{k=1}^{n}\theta_{k}(\delta_{k}-1).
\end{equation*}
There exists unique constant equilibrium, given by 
\begin{equation*}
\pi ^{i,*} = \left( \delta _{i} - \frac{\theta_{i}(\delta _{i}-1)\overline{\delta }}{1+\overline{\theta (\delta -1)}}\right)\frac{\mu }{\sigma ^{2}}.
\end{equation*}
\end{corollary}
\begin{proof}
Apply Theorem \ref{th:power-eq-np}, taking note of the simplifications $\varphi_n = \overline{\delta}\mu/\sigma$ and $\psi_n = \overline{\theta (\delta -1)}$.
\end{proof}

\begin{remark}
As in Remark \ref{re:alternative-averaging-exp} in the exponential utility model, one might modify our payoff structure so that agent $i$ excludes herself from the geometric mean $\overline{X}_T$. That is, one might replace the payoff functional $J_i$ defined in \eqref{def:Jpower} by
\[
\E\left[ U\left(X_{T}^{i}(\overline{X}^{(-i)}_T)^{-\theta'_i};\delta'_i\right)\right] ,\quad \text{where} \quad \overline{X}^{(-i)}_{T}=\left(\prod_{k\neq i}X_{T}^{k}\right) ^{1/(n-1)},
\]
for some parameters $\theta'_i \in [0,1]$ and $\delta'_i > 0$.
By modifying the preference parameters, we may view this payoff as a special case of ours.
Indeed, by matching coefficients it is straightforward to show that
\[
U\left(X_{T}^{i}(\overline{X}^{(-i)}_T)^{-\theta'_i};\delta'_i\right) = c_iU\left(X_{T}^{i}\overline{X}_T^{-\theta_i};\delta_i\right),
\]
for some constant $c_i > 0$ (which does not influence the optimal strategies), when $\theta_i \in [0,1]$ and $\delta_i > 0$ are defined by 
\[
\delta_i = \frac{\delta'_i}{\delta'_i - (\delta'_i-1)\left(1 + \frac{1}{n-1}\theta'_i\right)} \quad\quad \text{ and } \quad\quad \theta_i = \frac{\theta'_i}{\frac{n-1}{n} + \frac{1}{n}\theta'_i}.
\]
However, this is only valid if $(1-1/\delta'_i)\left(1 + \frac{1}{n-1}\theta'_i\right) < 1$, which ensures that $\delta_i > 0$. This certainly holds for sufficiently large $n$.
We favor our original parametrization because of the relative simplicity of the formulas in Theorem \ref{th:power-eq-np} and Corollary \ref{co:power-eq-np}, and because there is no difference in the $n \rightarrow \infty$ limit.
\end{remark}

\begin{proof}[Proof of Theorem \ref{th:power-eq-np}]
The proof is similar to that of Theorem \ref{th:exp-eq-np}, so we only highlight the main steps. Fix  an agent $i$ and constant strategies $\alpha_k \in \R$, for $k \neq i$. Define 
\begin{equation*}
Y_{t}:=\left( \prod_{k\neq i}X_{t}^{k}\right) ^{1/n},
\end{equation*}
where $X_{t}^{k}$ solves \eqref{def:X-i-power} with constant weights $\alpha_k$ and $X_{0}^{k}=x^{k}_0$.

Setting $\Sigma _{k}:=\sigma _{k}^{2}+\nu _{k}^{2}$, we deduce that 
\begin{equation*}
d\left( \log X_{t}^{k}\right) =\left( \mu _{k}\alpha _{k}-\frac{1}{2}\Sigma _{k}\alpha _{k}^{2}\right) dt+\nu _{k}\alpha _{k}dW_{t}^{k}+\sigma
_{k}\alpha _{k}dB_{t}.
\end{equation*}
In turn, 
\begin{equation*}
d\left( \log Y_{t}\right) =\frac{1}{n}\sum_{k\neq i}d\log X_{t}^{k}=\left( \widehat{\mu \alpha }-\frac{1}{2}\widehat{\Sigma \alpha ^{2}}\right) dt+\frac{1}{n}\sum_{k\neq i}\nu _{k}\alpha _{k}dW_{t}^{k}+\widehat{\sigma\alpha }dB_{t},
\end{equation*}
where we abbreviate
\begin{equation*}
\widehat{\mu \alpha }:=\frac{1}{n}\sum_{k\neq i}\mu _{k}\alpha _{k}, \quad\quad\quad\quad\quad \widehat{\sigma \alpha }:=\frac{1}{n}\sum_{k\neq i}\sigma
_{k}\alpha _{k},
\end{equation*}
\begin{equation*}
\widehat{\Sigma \alpha ^{2}}:=\frac{1}{n}\sum_{k\neq i}\Sigma \alpha_{k}^{2}  \quad\quad \text{and} \quad\quad  \widehat{(\nu \alpha )^{2}}:=\frac{1}{n}\sum_{k\neq i}\nu _{k}^{2}\alpha _{k}^{2}.
\end{equation*}
Thus, the process $Y_{t}$ solves
\begin{align}
\frac{dY_{t}}{Y_{t}}=\eta dt+\frac{1}{n}\sum_{k\neq i}\nu _{k}\alpha
_{k}dW_{t}^{k}+\widehat{\sigma \alpha }dB_{t}, \quad Y_0 = \left(\prod_{k\neq i}x^k_0\right)^{1/n}, \label{def:Ysde}
\end{align}
with
\begin{equation*}
\eta :=\widehat{\mu \alpha }-\frac{1}{2}\left( \widehat{\Sigma \alpha ^{2}}-\widehat{\sigma \alpha }^{2}-\frac{1}{n}\widehat{(\nu \alpha )^{2}}\right) .
\end{equation*}
The $i^{\text{th}}$ agent then solves the optimization problem 
\begin{equation}
\sup_{\pi^i \in \A}\E\left[U\left((X_{T}^i)^{1-\theta _{i}/n}Y_{T}^{-\theta_{i}};\delta_i\right)\right], \label{pf:power-nplayer-value}
\end{equation}
where
\begin{equation*}
dX^i_{t} = \pi^i_t X^i_t(\mu _{i}dt+\nu _{i}dW_{t}^{i}+\sigma _{i}dB_{t}), \ X^i_0 = x^i_0,
\end{equation*}
with $(Y_t)_{t \in [0,T]}$ solving \eqref{def:Ysde}.
We then obtain that the value \eqref{pf:power-nplayer-value} is equal to $v(X^i_0,Y_0,0)$, where $v(x,y,t)$ solves the HJB equation 
\begin{equation}
v_{t}+\max_{\pi \in \R}\left( \frac{1}{2}(\sigma _{i}^{2}+\nu _{i}^{2})\pi
^{2}x^{2}v_{xx}+\pi \left( \mu _{i}xv_{x}+\sigma _{i}\widehat{\sigma \alpha }
xyv_{xy}\right) \right) \label{def:np-power-HJB}
\end{equation}
\begin{equation*}
+\frac{1}{2}\left( \widehat{\sigma \alpha }^{2}+\frac{1}{n}\widehat{(\nu
\alpha )^{2}}\right) y^{2}v_{yy}+\eta yv_{y}=0,
\end{equation*}
for $(x,y,t) \in \R_+ \times \R_+ \times [0,T]$, with terminal condition
\[
v(x,y,T)=U(x^{1-\theta _{i}/n}y^{-\theta_{i}};\delta _{i}).
\]
Applying the first order conditions, the maximum in \eqref{def:np-power-HJB} is attained by
\begin{equation}
\pi ^{i,*}(x,y,t)=-\frac{\mu _{i}xv_{x}(x,y,t)+\sigma _{i}\widehat{\sigma \alpha }xyv_{xy}(x,y,t)}{(\sigma _{i}^{2}+\nu_{i}^{2})x^{2}v_{xx}(x,y,t)}.  \label{pi-power-n}
\end{equation}
In turn, equation \eqref{def:np-power-HJB} reduces to
\begin{equation}
v_{t}-\frac{1}{2}\frac{(\mu _{i}xv_{x}+\sigma _{i}\widehat{\sigma \alpha }xyv_{xy})^{2}}{(\sigma _{i}^{2}+\nu _{i}^{2})x^{2}v_{xx}}+\frac{1}{2}\left( \widehat{\sigma \alpha }^{2}+\frac{1}{n}\widehat{(\nu \alpha )^{2}}\right)y^{2}v_{yy}+\eta yv_{y}=0. \label{def:np-power-HJB-reduced}
\end{equation}
Working as in the proof of Theorem \ref{th:exp-eq-np}, we deduce that
the above HJB  equation has a unique smooth solution (in an appropriate class of time-separable and space-homogeneous solutions), and the optimal feedback control in \eqref{pi-power-n} reduces to
\begin{align}
\pi ^{i,*}=\frac{\delta _{i}\mu _{i} - \sigma _{i}\widehat{\sigma \alpha }\theta _{i}(\delta _{i}-1)}{(\sigma _{i}^{2}+\nu_{i}^{2})(\delta _{i}-(1-\theta _{i}/n)(\delta _{i}-1))}. \label{pi-power-n-solved}
\end{align}
We prove this in two cases:

\begin{enumerate}[(i)]
\item Suppose $\delta _{i}\neq 1.$
Making the ansatz 
\begin{equation*}
v(x,y,t)=U(x^{1-\theta _{i}/n}y^{-\theta _{i}};\delta _{i})f(t)=(1-1/\delta_{i})^{-1}(x^{(1-\theta _{i}/n)}y^{-\theta _{i}})^{1-1/\delta_{i}}f(t)
\end{equation*}
reduces equation \eqref{def:np-power-HJB-reduced} to $(1-1/\delta_{i})^{-1}f'(t)+\rho f(t)=0$, 
with $f(T) =1$, where
\begin{equation*}
\rho :=\, \, \frac{(\mu _{i}(1-\theta _{i}/n)-\sigma _{i}\widehat{\sigma\alpha }\theta _{i}(1-\theta _{i}/n)(1-1/\delta_{i}))^{2}}{2(\sigma_{i}^{2}+\nu _{i}^{2})(1-\theta _{i}/n)(1-(1-\theta _{i}/n)(1-1/\delta_{i}))}-\eta \theta_{i}
\end{equation*}
\begin{equation*}
\,+\frac{1}{2}\left( \widehat{\sigma \alpha }^{2}+\frac{1}{n}\widehat{(\nu\alpha )^{2}}\right) \theta _{i}(1+\theta _{i}(1-1/\delta_{i})).
\end{equation*}
We easily deduce that the solution of \eqref{def:np-power-HJB-reduced} is
\[
v(x,y,t) =(1-1/\delta_{i})^{-1}(x^{(1-\theta_{i}/n)}y^{-\theta _{i}})^{1-1/\delta_{i}}e^{\rho (1-1/\delta_{i})(T-t)},
\]
and that \eqref{pi-power-n} yields \eqref{pi-power-n-solved}.
\item Suppose $\delta _{i} = 1$. 
Making the ansatz 
\begin{equation*}
v(x,y,t) = U(x^{1-\theta _{i}/n}y^{-\theta _{i}};\delta _{i}) + f(t) = \left(1-\frac{\theta _{i}}{n}\right)\log x-\theta_{i}\log y+f(t)
\end{equation*}
reduces equation \eqref{def:np-power-HJB-reduced} to $f'(t) + \rho = 0$, with $f(T)=0$ and
\begin{equation*}
\rho :=\frac{\mu _{i}^{2}(1-\theta _{i}/n)}{2(\sigma _{i}^{2}+\nu _{i}^{2})}-\theta _{i}\eta +\frac{1}{2}\theta _{i}\left( \widehat{\sigma \alpha }^{2}+\frac{1}{n}\widehat{(\nu \alpha )^{2}}\right) .
\end{equation*}
In turn, the solution of \eqref{def:np-power-HJB-reduced} is given by
\begin{equation*}
v(x,y,t) = \left(1-\frac{\theta _{i}}{n}\right)\log x-\theta _{i}\log y+\rho(T-t),
\end{equation*}
and \eqref{pi-power-n} reduces to $\pi^{i,*}=\mu_i/(\sigma_i^2+\nu_i^2)$, which is consistent with \eqref{pi-power-n-solved} for $\delta_i=1$.
\end{enumerate}

We conclude the proof as follows. For $(\alpha_1,\ldots ,\alpha_n)$ to be a constant equilibrium, we must have $\pi^{i,* }=\alpha_i,$ for each $i=1,\ldots,n$. Using \eqref{pi-power-n-solved} and abbreviating 
\begin{equation*}
\overline{\sigma \alpha }:=\frac{1}{n}\sum_{k=1}^{n}\sigma _{k}\alpha _{k}=\widehat{\sigma \alpha }+\frac{1}{n}\sigma _{i}\alpha _{i},
\end{equation*}
we deduce that we must have 
\begin{equation*}
\alpha_i =\frac{\mu_{i}-\sigma_{i}\overline{\sigma \alpha }\theta_{i}(1-1/\delta_i)+\sigma_{i}^{2}\alpha_i(\theta_{i}/n)(1-1/\delta_i)}{(\sigma_{i}^{2}+\nu_{i}^{2})(1-(1-\theta_{i}/n)(1-1/\delta_i))}.
\end{equation*}
Solving for $\alpha_i$ yields,
\begin{align}
\alpha_i &=\frac{\mu _{i}-\sigma _{i}\overline{\sigma \alpha }\theta
_{i}(1-1/\delta_{i})}{(\sigma _{i}^{2}+\nu _{i}^{2})(1-(1-\theta
_{i}/n)(1-1/\delta_{i}))}\left( 1-\frac{\sigma _{i}^{2}(\theta
_{i}/n)(1-1/\delta_{i})}{(\sigma _{i}^{2}+\nu _{i}^{2})(1-(1-\theta
_{i}/n)(1-1/\delta_{i}))}\right) ^{-1} \nonumber \\
	&=\frac{\mu _{i}-\sigma _{i}\overline{\sigma \alpha }\theta _{i}(1-1/\delta_{i})}{(\sigma _{i}^{2}+\nu _{i}^{2})(1-(1-\theta _{i}/n)(1-1/\delta_{i}))-\sigma _{i}^{2}(\theta _{i}/n)(1-1/\delta_{i})} \nonumber \\
	&=\frac{\mu _{i}-\sigma _{i}\overline{\sigma \alpha }\theta _{i}(1-1/\delta_{i})}{\sigma _{i}^{2}/\delta_{i} + \nu _{i}^{2}(1-(1-\theta
_{i}/n)(1-1/\delta_{i}))} =\frac{\mu _{i}\delta _{i}-\sigma _{i}\overline{\sigma \alpha }\theta _{i}(\delta _{i}-1)}{\sigma _{i}^{2}+\nu_{i}^{2}(1-\theta _{i}/n+\delta _{i}\theta _{i}/n)}. \label{def:alpha-expression-power}
\end{align}
Multiplying both sides by $\sigma _{i}$ and averaging over $i=1,\ldots ,n$ give
\begin{equation}
\overline{\sigma \alpha }=\varphi _{n}-\psi _{n}\overline{\sigma \alpha }, \label{eq-eqn-power}
\end{equation}
where $\varphi _{n},\psi _{n}$ are as in (\ref{phi-n-power}) and \eqref{psi-n-power}. Because $1+\psi_n > 0$, equation \eqref{eq-eqn-power} holds if and only if $\overline{\sigma \alpha }=\varphi _{n}/(1+\psi_{n})$. We then deduce from \eqref{def:alpha-expression-power} that the equilibrium strategy $\alpha_i=\pi^{i,*}$ is given by \eqref{eq-power}.
\end{proof}

\begin{remark}
Note that equation \eqref{def:alpha-expression-power} above has a \emph{unique} solution for \emph{all} parameter values. In contrast, the analogous equation \eqref{eqn-aux} in the exponential case has no solutions for certain parameter values, which is why there were two cases in Theorem \ref{th:exp-eq-np}.
\end{remark}

It is worth highlighting that, for the CRRA case, we assume that relative performance concerns appear multiplicatively, and not additively. There are two reasons for this. First, as discussed in the introduction, this is natural in modeling preferences which depend on relative \emph{return} as opposed to relative \emph{wealth}; see also \cite{basak2015competition} for a discussion. The second reason is mathematical tractability. We have already seen that using the geometric mean leads to explicit solutions.

To formulate an analogous problem using an arithmetic mean, we may consider the following two possibilities. First, we may modify the optimization criterion to be of the form
\[
U\left(\frac{X^i_T}{\frac{1}{n}\sum_{i=1}^nX^i_T}; \, \delta_i\right).
\]
The challenge here is that the ratio appearing in the first argument cannot be expressed as the solution of a one-dimensional SDE. The proofs of both Theorems \ref{th:power-eq-np} and \ref{th:power-eq} exploited the fact that the geometric mean of geometric Brownian motions remains a geometric Brownian motion, whereas the arithmetic mean enjoys no such properties.

Alternatively, we may use an optimization criterion of the form $U(X^i_T - \frac{1}{n}\sum_{i=1}^nX^i_T;\delta_i)$, which, however, runs into more serious problems because $U(x,\delta_i)$ is well-defined and finite only for $x > 0$ (or $x \ge 0$ if $\delta_i > 1$). Hence, this criterion would enforce the hard constraint $X^i_T > \frac{1}{n}\sum_{i=1}^nX^i_T$ a.s., which raises the two natural questions of how this constraint propagates to previous times and whether this results in a meaningful class of solutions.

In short, using an arithmetic mean criterion would give rise to inter-dependent state and control constraints which will likely render the problem intractable, and, at worst, could lead to trivial or meaningless solutions.

\subsection{The mean field game} \label{se:mfg-power}

This section studies the limit as $n\rightarrow\infty$ of the $n$-player game analyzed in the previous section, analogously to the treatment of the exponential case in Section \ref{se:mfg-exp}.

We proceed with some informal arguments.
Recall that the \emph{type vector} of agent $i$ is
\begin{equation*}
\zeta_i := (x^i_0,\delta_{i},\theta_{i},\mu_{i},\nu_{i},\sigma_{i}).
\end{equation*}
As before, the type vectors induce an empirical measure, which is the probability measure on
\begin{align}
\Z^p  := (0,\infty) \times (0,\infty) \times [0,1] \times (0,\infty) \times [0,\infty) \times [0,\infty) \label{def:power-type-space}
\end{align}
given by
\[
m_n(A) = \frac{1}{n}\sum_{i=1}^n1_A(\zeta_i), \ \text{ for Borel sets } A \subset \Z^p .
\]
Similarly to the exponential case, for a given agent $i$, the equilibrium strategy $\pi^{i,*}$ computed in Theorem \ref{th:power-eq-np} depends only on her own type vector $\zeta_i$ and the distribution $m_n$ of type vectors, and this enables the passage to the limit.

Assume now that $m_n$ has a weak limit $m$, in the sense that $\int_{\Z^p } f\,dm_n \rightarrow \int_{\Z^p } f\,dm$ for every bounded continuous function $f$ on $\Z^p $. Let $\zeta=(\xi,\delta,\theta,\mu,\nu,\sigma)$ denote a random variable with distribution $m$. Then, the optimal strategy  $\pi^{i,*}$ (cf. \eqref{eq-power}) should converge to
\begin{align}
\lim_{n\rightarrow\infty}\pi^{i,*} =\delta_i\frac{\mu_i}{\sigma_i^2 + \nu_i^2}-\theta_i(\delta_i-1)\frac{\sigma _i}{\sigma_i^2 + \nu_i^2}\frac{\varphi}{1-\psi}, \label{def:power-limitingstrategy}
\end{align}
where
\begin{align*}
\varphi := \lim_{n\uparrow \infty }\varphi _{n}= \E\left[ \delta\frac{\mu \sigma  }{\sigma ^{2}+\nu ^{2}}\right] \quad\quad \text{and} \quad\quad \psi := \lim_{n\uparrow \infty }\psi _{n} =\E\left[ \theta (\delta -1)\frac{\sigma ^{2}}{\sigma ^{2}+\nu ^{2}}\right].
\end{align*}
As in the exponential case, we will demonstrate that this limiting strategy is indeed the equilibrium of a mean field game, which we formulate analogously to Section \ref{se:mfg-exp-formulation}.

Recall that $W$ and $B$ are independent Brownian motions and that the random variable $\zeta = (\xi,\delta ,\theta ,\mu ,\nu ,\sigma)$ is independent of $W$ and $B$. For the power case, the type vector $\zeta$ now takes values in the space $\Z^p $.
Furthermore, the filtration $\FF^{\mathrm{MF}}$ is the smallest one satisfying the usual assumptions for which $\zeta $ is $\F^{\mathrm{MF}}_0$-measurable and $W$ and $B$ are adapted. Finally, recall that $\FF^B=(\F^B_t)_{t \in [0,T]}$ denotes the natural filtration generated by the Brownian motion $B$.

The \emph{representative agent's} wealth process solves
\begin{align}
dX_t = \pi_tX_t(\mu dt + \nu dW_t + \sigma dB_t), \ X_0 = \xi, \label{def:X-MFG-power}
\end{align}
where the investment weight $\pi$ belongs to the admissible set $\A_{\mathrm{MF}}$ of $\FF^{\mathrm{MF}}$-progressively measurable real-valued processes satisfying $\E\int_0^T|\pi_t|^2dt < \infty$. Notice that, for all admissible $\pi$, the wealth process $(X_t)_{t \in [0,T]}$ is strictly positive, as $\xi > 0$ a.s.

We denote by $\overline{X}$ an $\F^{\mathrm{MF}}_0$-measurable random variable representing the geometric mean wealth among the continuum of agents. Then, the objective of the representative agent is to maximize the expected payoff
\begin{align}
\sup_{\pi \in \A_{\mathrm{MF}}}\E\left[U(X_T\overline{X}^{-\theta};\delta)\right], \label{def:power-mfg-optimization}
\end{align}
where $(X_t)_{t \in [0,T]}$ is given by \eqref{def:X-MFG-power}.

The definition of a mean field equilibrium is analogous to Definition \ref{def:MFE-exp}. However, one needs to extend the notion of geometric mean appropriately to a continuum of agents. The geometric mean of a measure $m$ on $(0,\infty)$ is most naturally defined as
\[
\exp\left(\int_{(0,\infty)}\log y\, \,dm(y)\right),
\]
when $\log y$ is $m$-integrable. Indeed, when $m$ is the empirical measure of $n$ points $(y_1,\ldots,y_n)$, this reduces to the usual definition $(y_1y_2\cdots y_n)^{1/n}$.

\begin{definition} \label{def:MFE-power}
Let $\pi^* \in \A_{\mathrm{MF}}$ be an admissible strategy, and consider the $\F^B_T$-measurable random variable $\overline{X} := \exp \E[\log X^*_T \,|\, \F^B_T]$, where $(X^*_t)_{t \in [0,T]}$ is the wealth process in \eqref{def:X-MFG-power} corresponding to the strategy $\pi^*$. We say that $\pi^*$ is a mean field equilibrium (MFE) if $\pi^*$ is optimal for the optimization problem \eqref{def:power-mfg-optimization} corresponding to this choice of $\overline{X}$.

A constant MFE is an $\F^{\mathrm{MF}}_0$-measurable random variable $\pi^*$ such that, if $\pi_t := \pi^*$ for all $t \in [0,T]$, then $(\pi_t)_{t \in [0,T]}$ is a MFE. 
\end{definition}

The following theorem characterizes the constant MFE, recovering the limiting expressions derived above from the $n$-agent equilibria.

\begin{theorem} \label{th:power-eq}
Assume that, a.s., $\delta > 0$, $\theta \in [0,1]$, $\mu > 0$, $\sigma \ge 0$, $\nu \ge 0$, and $\sigma+\nu > 0$. Define the constants 
\begin{equation*}
\varphi := \E\left[ \delta\frac{\mu \sigma  }{\sigma^{2} + \nu^{2}}\right] \quad \quad \text{ and } \quad \quad \psi := \E\left[\theta (\delta -1) \frac{\sigma ^{2}}{\sigma^{2} + \nu^{2}}\right] ,
\end{equation*}
where we assume that both expectations exist and are finite.

There exists a unique constant MFE, given by 
\begin{equation}
\pi^* = \delta\frac{\mu}{\sigma ^{2}+\nu ^{2}} - \theta (\delta -1)\frac{\sigma}{\sigma ^{2}+\nu ^{2}}\frac{ \varphi }{1+\psi }. \label{def:power-eq-control}
\end{equation}
Moreover, we have the identity
\begin{align*}
\E[\sigma\pi^*] = \frac{\varphi}{1+\psi}.
\end{align*}
\end{theorem}

In the single stock case, the form of the solution is essentially the same as in the $n$-agent game, presented in Corollary \ref{co:power-eq-np}:

\begin{corollary}[Single stock] \label{co:power-eq-single}
Suppose $(\mu,\nu,\sigma)$ are deterministic, with $\nu = 0$ and $\mu,\sigma > 0$. Define the constants
\begin{equation*}
\overline{\delta } := \E[\delta] \quad \quad \text{ and } \quad \quad \overline{\theta(\delta-1)} := \E[\theta(\delta-1)].
\end{equation*}
There exists a unique constant MFE, given by 
\begin{equation*}
\pi^* =\left( \delta - \frac{\theta(\delta - 1) \overline{\delta}}{1+\overline{\theta(\delta-1)}}\right)\frac{\mu }{\sigma ^{2}}.
\end{equation*}
\end{corollary}

\begin{proof}[Proof of Theorem \ref{th:power-eq}]
\bigskip As in the exponential case, we first reduce the optimal control
problem \eqref{def:power-mfg-optimization} to a low-dimensional Markovian one. To this end, it suffices to restrict our attention to random variables $\overline{X}$ of the form 
\begin{equation*}
\overline{X}=\exp \E[\log X_{T}^{\alpha }| \F^B_T],
\end{equation*}
where $X^\alpha$ is the wealth process of \eqref{def:X-MFG-power} with an admissible constant strategy  $\alpha$. That is, $\alpha$ is an $\F^{\mathrm{MF}}_0$-measurable random variable satisfying $\E[\alpha^2] < \infty$.

Define
\begin{equation*}
Y_t := \exp \E[\log X_t^{\alpha } | \F^B_T].
\end{equation*}
Note that $Y_t = \exp \E[\log X_t^{\alpha } | \F^B_t]$, for $t \in [0,T)$, because $(B_s-B_t)_{s \in [t,T]}$ and $X^\alpha_t$ are independent. In analogy to the exponential case, we identify the dynamics of $Y$ and, in
turn, treat it, as an additional (uncontrolled) state process.

To this end, first use It\^{o}'s formula to get 
\begin{equation*}
d\left( \log X_{t}^{\alpha }\right) =\left( \mu \alpha -\frac{1}{2}(\sigma^{2}+\nu^{2})\alpha ^{2}\right) dt+\nu \alpha dW_{t}+\sigma \alpha
dB_{t}.
\end{equation*}
Define $\check{X}^\alpha_{t}:=\E[\log X_{t}^{\alpha }|\F^B_T]$, and note as with $Y_t$ above that $\check{X}^\alpha_{t}=\E[\log X_{t}^{\alpha }|\F^B_t]$, for $t \in [0,T)$. Setting
\begin{equation*}
\Sigma := \sigma ^{2} + \nu ^{2},
\end{equation*}
and noting that $(\xi ,\mu ,\sigma ,\nu ,\alpha )$, $W$ and $B$ are independent, we compute
\begin{equation}
d\check{X}^\alpha_{t}=\left( \overline{\mu \alpha }-\frac{1}{2}\overline{\Sigma
\alpha ^{2}}\right) dt+\overline{\sigma \alpha }dB_{t}, \label{def:check-X}
\end{equation}
where, again, we use the notation $\overline{M}=\E[M]$ for a generic integrable random variable $M$. In turn,
\begin{align}
dY_t = de^{\check{X}^\alpha_{t}}=Y_t\left( \eta dt+\overline{\sigma\alpha }dB_{t}\right) , \quad Y_0 = \overline{\xi}, \label{def:Ysde-power}
\end{align}
where $\eta :=\overline{\mu \alpha }-\frac{1}{2}(\overline{\Sigma \alpha ^{2}} - \overline{\sigma \alpha }^{2})$.

To solve the stochastic optimization problem \eqref{def:power-mfg-optimization}, we equivalently solve
\begin{equation}
\sup_{\pi \in \A_{\mathrm{MF}}}\E\left[ U(X_TY_T^{-\theta};\delta)\right] \label{pf:power-mfg-optimization-reduced}
\end{equation}
with
\begin{equation*}
dX_{t}=\pi _{t}X_{t}(\mu dt+\nu dW_{t}+\sigma dB_{t}),
\end{equation*}
and $(Y_t)_{t \in [0,T]}$ solving \eqref{def:Ysde-power}.
Then, as in the discussion of Section \ref{se:alternative-mfg}, the value of \eqref{pf:power-mfg-optimization-reduced} is equal to $\E[v(\xi,\overline{\xi},0)]$, where $v=v(x,y,t)$ is the unique smooth (strictly concave and strictly increasing in $x$) solution of the HJB equation
\begin{equation}
v_{t} + \max_{\pi \in \R}\left( \frac{1}{2}\Sigma \pi ^{2}x^{2}v_{xx}+\pi \left( \mu xv_{x}+\sigma \overline{\sigma \alpha }xyv_{xy}\right) \right) 
+\frac{1}{2}\overline{\sigma \alpha }^{2}y^{2}v_{yy}+\eta yv_{y}=0, \label{def:power-HJB}
\end{equation}
with terminal condition $v(x,y,T)=U(xy^{-\theta };\delta)$.
Notice that this HJB equation is random, because of its dependence on the $\F^{\mathrm{MF}}_0$-measurable type parameters.

Applying the first order conditions, the maximum in \eqref{def:power-HJB} is attained by 
\begin{align}
\pi^*(x,y,t) &= -\frac{\mu x v_x(x,y,t) + \sigma\overline{\sigma\alpha}xy v_{xy}(x,y,t)}{\Sigma x^2 v_{xx}(x,y,t)}. \label{pf:power-maximizes-hamiltonian}
\end{align}
In turn, equation \eqref{def:power-HJB} reduces to
\begin{align}
v_{t}-\frac{(\mu xv_{x}+\sigma \overline{\sigma \alpha }xyv_{xy})^{2}}{2\Sigma x^{2}v_{xx}} + \frac{1}{2}\overline{\sigma \alpha }^{2}y^{2}v_{yy} + \eta yv_{y}=0.\label{pf:power-HJB-reduced}
\end{align}
Next, we claim that, for all $(x,y,t)$,
\begin{align}
\pi^*(x,y,t) &= \Sigma^{-1}\left(\mu\delta - \theta(\delta-1)\sigma\overline{\sigma\alpha}\right). \label{pf:power-optimal-control}
\end{align}
We prove this in two cases:

\begin{enumerate}[(i)]
\item Suppose $\delta\neq 1$.
Making the ansatz
\begin{align*}
v(x,y,t) = U(xy^{-\theta};\delta)f(t) = \left(1-1/\delta\right)^{-1}x^{1-1/\delta}y^{-\theta(1-1/\delta)}f(t),
\end{align*}
reduces equation \eqref{pf:power-HJB-reduced} to $f'(t) + \rho f(t)=0$, with $f(T)=1$, where
\begin{equation*}
\rho :=\frac{\left( \mu \delta -\theta (\delta -1)\sigma \overline{\sigma\alpha }\right) ^{2}}{2\Sigma (\delta -1)}-\eta \theta (1-1/\delta)^{-1}+\frac{1}{2}\overline{\sigma \alpha }^{2}\theta (\theta +(1-1/\delta)^{-1}).
\end{equation*}
We easily deduce that the solution of \eqref{pf:power-HJB-reduced} is
\[
v(x,y,t) = \left(1-1/\delta\right)^{-1}x^{1-1/\delta}y^{-\theta(1-1/\delta)}\exp(\rho(T-t)),
\]
and that \eqref{pf:power-maximizes-hamiltonian} yields \eqref{pf:power-optimal-control}.
\item Suppose $\delta = 1$.
It is easily checked that the solution $v$ of \eqref{pf:power-HJB-reduced} is given by
 \begin{equation*}
v(x,y,t)=\log x-\theta \log y + \rho(T-t) ,
\end{equation*}
with 
\begin{equation*}
\rho :=\frac{\mu ^{2}}{2\Sigma }-\eta \theta +\frac{1}{2}\theta \overline{\sigma \alpha }^{2}.
\end{equation*}
In this case, \eqref{pf:power-maximizes-hamiltonian} becomes $\pi^*(x,y,t) = \mu/\Sigma$, which is consistent with \eqref{pf:power-optimal-control} for $\delta=1$.
\end{enumerate}

Recalling Definition \ref{def:MFE-power}, we see that for the candidate control $\alpha$ to be a constant MFE, we need $\alpha = \pi^*$. In light of \eqref{def:MFE-power}, $\pi^*$ is a constant MFE if it solves the equation
\begin{align}
\pi^* &= \Sigma^{-1}\left(\mu\delta - \theta(\delta-1)\sigma\overline{\sigma \pi^*}\right). \label{pf:power-optimal-control2}
\end{align}
Multiply both sides by $\sigma$ and average to find that $\overline{\sigma \pi^*}$ must satisfy
\begin{align*}
\overline{\sigma \pi^*} = \E\left[\delta\frac{\mu\sigma}{\Sigma}\right] - \E\left[\theta(\delta-1)\frac{\sigma^2}{\Sigma}\right]\overline{\sigma \pi^*} = \varphi - \psi\overline{\sigma \pi^*}.
\end{align*}
We then deduce that $\overline{\sigma\alpha} = \varphi/(1+\psi)$, and plugging this into \eqref{pf:power-optimal-control} we obtain \eqref{def:power-eq-control}.
\end{proof}

\subsection{Discussion of the equilibrium} \label{se:discussion-power}
Some of the structural properties of the equilibrium are similar to those observed in the CARA model  in Section \ref{se:discussion-exp}.
We again focus the discussion here on the mean field case of Theorem \ref{th:power-eq} and Corollary \ref{co:power-eq-single}, as the $n$-agent equilibria of Theorem \ref{th:power-eq-np} and Corollary \ref{co:power-eq-np} have essentially the same structure. The only difference is the rescaling of $\nu_k^2$ by $(1+(\delta_k-1)\theta_k/n)$ wherever it appears in Theorem \ref{th:power-eq-np}. We first discuss the general case of Theorem \ref{th:power-eq}, before concentrating on the single stock case of Corollary \ref{co:power-eq-single}.

\subsubsection{The general case}

The MFE $\pi^*$ of Theorem \ref{th:power-eq} can be written as the sum of two components, $\pi^* = \pi^*_{(1)}+\pi^*_{(2)}$, where $\pi^*_{(1)} = \delta\mu/(\sigma^2+\nu^2)$ is the classical Merton portfolio, and 
\[
\pi^*_{(2)} := - \theta (\delta -1)\frac{\sigma}{\sigma ^{2}+\nu ^{2}}\frac{ \varphi }{1+\psi }.
\]
The second component $\pi^*_{(2)}$ isolates the linear effect of the competition parameter $\theta$. Notably, $\pi^*_{(2)}$ vanishes when $\theta=0$.

Interestingly, the effect of competition is quite different in the CRRA model than in the CARA model, in the sense that competition now leads some agents to invest \emph{less} in the risky asset than they would in the absence of competition. Indeed, the sign of $\pi^*_{(2)}$ is the same as that of $(1-\delta)$, assuming $\theta > 0$ and $\sigma > 0$. Thus, agents with $\delta < 1$ invest \emph{more} as $\theta$ increases, whereas agents with $\delta > 1$ invest \emph{less}. In particular, we have $\pi^*_{(2)}=0$ when $\delta=1$; that is, agents with log utility are not competitive, which is also easily deduced from the original problem formulation.

In fact, a highly risk-tolerant and competitive agent may choose to short the stock. That is, if $\delta > 1$ and $\theta$ is close to $1$, $\pi^*$ may be negative. This typically occurs when $\delta$ is much higher than their population averages, or, in other words, when the representative agent is very risk tolerant and competitive relative to the other agents.

The representative agent's strategy $\pi^*$ is influenced by the other agents only through the  quantity $\varphi/(1+\psi) = \E[\sigma\pi^*]$, and, as in Section \ref{se:discussion-exp}, we can  view this quantity as the volatility of the aggregate wealth. Indeed, let $X^*$ denote the wealth process corresponding to $\pi^*$ (i.e., the solution of \eqref{def:X-MFG-power} using the strategy $\pi^*$). The geometric average wealth of the population at time $t \in [0,T]$ is $Y_t := \log \E[\exp(X^*_t) | \F^B_T]$, and, as we saw in the proof of Theorem \ref{th:power-eq}, it satisfies
\[
dY_t = Y_t\left(\eta dt + \E[\sigma\pi^*]dB_t\right).
\]

Alternatively, the ratio $\varphi/(1+\psi)$ can be interpreted directly in terms of the type distribution.
Define $R = \sigma^2/(\sigma^2+\nu^2)$, and note that $\varphi = \E[R\delta\mu/\sigma]$ and $\psi = \E[R\theta(\delta-1)]$. Notice that the assumptions on the parameter ranges ensure that $1+\psi > 0$. As before, the numerator $\varphi$ increases as the quality of the other stocks increases, as measured by their Sharpe ratio. However, the ratio $\varphi/(1+\psi)$ may not increase as the population becomes more risk tolerant (i.e., as $\delta$ increases on average), as both the numerator and denominator increase in this case.

The dependence of $\varphi/(1+\psi)$ and thus of $\pi^*$ on the type distribution is rather complex. The distribution of competition weights $\theta$ appears only through $\psi$, and its effect is mediated by the risk tolerance $\delta$. Loosely speaking, the population average of $\theta$ can have a positive or negative effect on $\pi^*_{(2)}$ depending on the ``typical'' sign of $(1-\delta)$. These complexities are more easily unraveled in the single stock case.

\subsubsection{Single stock case}

From the results of Corollary \ref{co:power-eq-single}, we may write the equilibrium portfolio in the single stock case as
\begin{align}
\pi^* =\left( (1-k\theta)\delta + k\theta\right)\frac{\mu }{\sigma ^{2}}, \label{def:pi*-k}
\end{align}
where 
\begin{align}
k := \frac{\bar{\delta}}{1 + \overline{\theta(\delta-1)}}. \label{def:deltafactor}
\end{align}

The equilibrium $\pi^*$ can thus be written  as a Merton portfolio, $\pi^*=\delta_{\mathrm{eff}}\mu/\sigma^2$, with 
\emph{effective risk tolerance parameter} 
\begin{align}
\delta_{\mathrm{eff}} := (1-k\theta)\delta + k\theta = \delta -\frac{\theta (\delta -1)\overline{\delta}}{1+\overline{\theta (\delta -1)}}. \label{def:delta_eff}
\end{align}
This representation simplifies some of the complex dependencies of $\pi^*$ on the type distribution mentioned in the previous paragraph. For instance, suppose $\theta$ and $\delta$ are uncorrelated, so that $\overline{\theta(\delta-1)} = \overline{\theta}(\overline{\delta}-1)$. If $\overline{\delta} > 1$, then $|\delta_{\mathrm{eff}}-\delta|$ is decreasing in $\overline{\theta}$. That is, if the average risk tolerance is high, then, as the population becomes more competitive (i.e. $\overline{\theta}$ increases), the representative agent behaves less competitively in the sense that $\delta_{\mathrm{eff}}$ moves closer to $\delta$. On the other hand, if $\overline{\delta} < 1$, then $|\delta_{\mathrm{eff}}-\delta|$ is increasing in $\overline{\theta}$. That is, if the average risk tolerance is low, then, as the population becomes more competitive, the representative agent behaves more competitively in the sense that $\delta_{\mathrm{eff}}$ moves away from $\delta$. Again, if $\overline{\delta}=1$, then $\overline{\theta}$ and $\theta$ play no role whatsoever.

More interesting is the joint effect of $(\delta,\overline{\theta})$ on $\pi^*$, when the other parameters are fixed. Still assuming $\theta$ and $\delta$ are uncorrelated, notice that the value of $k$ can range between $1$ and $\overline{\delta}$, as $\overline{\theta}$ varies between $0$ and $1$. Hence, if $\theta\overline{\delta} > 1$, then there is a critical value, $\overline{\theta}_{\mathrm{crit}} := (\theta\overline{\delta}-1)/(\overline{\delta}-1)$, at which the effect of $\delta$ on $\pi^*$ changes sign. When the population is highly competitive (i.e., $\overline{\theta} > \overline{\theta}_{\mathrm{crit}}$), the investment $\pi^*$ in the risky asset increases with the risk tolerance $\delta$, as one might expect. On the other hand, when the population is less competitive (i.e., $\overline{\theta} < \overline{\theta}_{\mathrm{crit}}$), $\pi^*$ is decreasing in $\delta$. This effect is illustrated in Figure \ref{fig:CRRA-delta-vs-thetabar}.

\begin{figure}[h]
\centering
\includegraphics[scale=0.7]{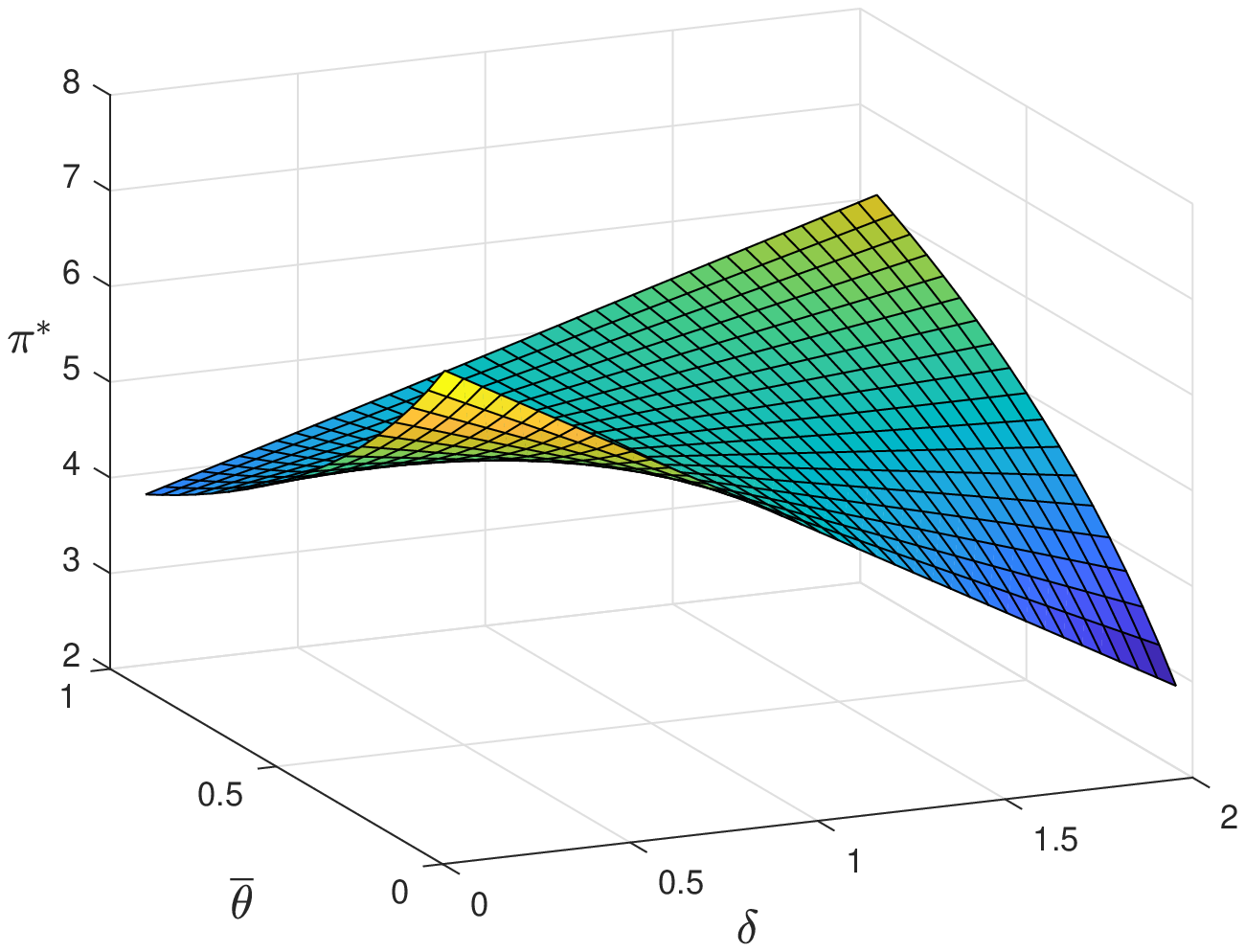}
\caption{Single stock case (Corollary \ref{co:power-eq-single}): $\pi^*$ versus $\delta$ and $\overline{\theta}$, with $\theta=3/4$, $\overline{\delta}=2$, $\mu=5$, and $\sigma=1$. Here, $\theta$ and $\delta$ are uncorrelated, and $\overline{\theta}_{\mathrm{crit}} = 1/2$}
\label{fig:CRRA-delta-vs-thetabar}
\end{figure}

A similar transition appears in the joint effect of $(\delta,\theta)$ on $\pi^*$, when the other parameters are fixed. When $k \le 1$ (which is equivalent to $\overline{\delta} \le 1$ if we assume $\theta$ and $\delta$ are uncorrelated), then the risky invesment $\pi^*$ is increasing in the risk tolerance $\delta$, for any value of $\theta$. On the other hand, if $k < 1$, then $\pi^*$ is increasing in $\delta$ if and only if $\theta < 1/k$. This situation is illustrated in Figure \ref{fig:CRRA-delta-vs-theta}. Note that these effects are more pronounced if $\delta$ and $\theta$ are positively correlated and less pronounced if they are negatively correlated.

\begin{figure}[h]
\centering
\includegraphics[scale=0.7]{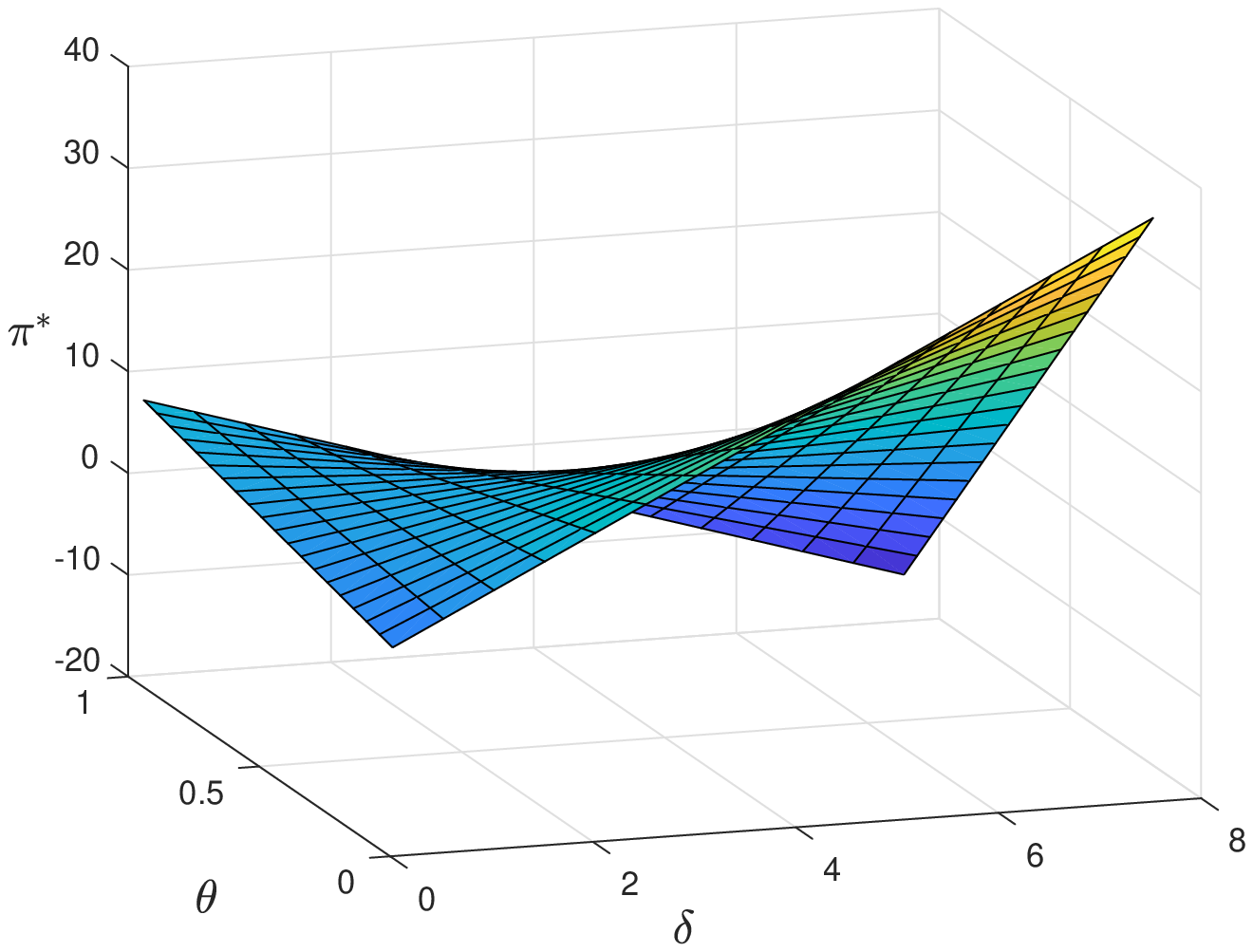}
\caption{Single stock case (Corollary \ref{co:power-eq-single}): $\pi^*$ versus $\delta$ and $\theta$, with $\overline{\theta}=1/5$, $\overline{\delta}=2$, and $\mu=5$, $\sigma=1$. Here $\delta$ and $\theta$ are uncorrelated, and $k = 5/3$.}
\label{fig:CRRA-delta-vs-theta}
\end{figure}

There are two ways to explain the counter-intuitive phenomenon described in the previous two paragraphs, in which $\pi^*$ may be decreasing in $\delta$ for certain (fixed) values of the other parameters. This regime happens precisely when $1 < k\theta$. Assuming again that $\theta$ and $\delta$ are uncorrelated, the latter inequality is equivalent to
\begin{align*}
\theta > \frac{1}{\overline{\delta}}(1 - \overline{\theta}) + \overline{\theta}.
\end{align*}
Assume for the sake of argument that $\overline{\delta}$ is extremely large. In the limit $\overline{\delta} \rightarrow\infty$, we see that $1 < k\theta$ if and only if  $\theta > \overline{\theta}$, and the expression for $\pi^*$ becomes
\begin{align*}
\pi^* =\left( \left(1-\frac{\theta}{\overline\theta}\right)\delta + \frac{\theta}{\overline\theta}\right)\frac{\mu }{\sigma ^{2}}.
\end{align*}
There are two immediate observations:
\begin{enumerate}[(i)]
\item If $\theta < \overline\theta$ and $\delta$ is large, then $\pi^*$ is positive and large. That is, less competitive agents go long.
\item If $\theta > \overline\theta$ and $\delta$ is large, then $\pi^*$ is very negative. That is, competitive agents go short.
\end{enumerate}
The first case, $\theta < \overline{\theta}$, is natural: A less competitive agent behaves more like a Merton investor, with $\pi^*$ increasing in $\delta$. On the other hand, we may explain the second regime, $\theta > \overline\theta$, as follows. Still assuming that the average risk tolerance is very large, we know from (i) that most other agents will take large long positions in the stock. If our representative agent were to go long, he could not afford to accept enough risk to go \emph{as long} as the other investors. Then, if the stock price goes up, the representative agent may achieve a high absolute performance but would suffer in terms of relative performance, as the other agents who invested even more in the stock would earn even higher returns. Hence, a natural strategy is to short the stock, to focus on gains from outperforming competitors when the stock price drops, rather than focusing on absolute performance. This is still reasonable if the representative agent is himself quite risk tolerant, willing to accept risk in the opposite direction as the more Merton-like investors. Note that these effects are less pronounced but still present outside of the asymptotic regime $\overline{\delta} \rightarrow \infty$. 

There is an alternative explanation in the spirit of \cite[pp. 13-14]{basak2015competition}. A risk averse agent typically seeks to minimize volatility by investing less in the stock than a more risk tolerant agent. However, relative performance concerns provide an additional source of volatility. A risk averse agent may then invest heavily in the stock in an attempt to mitigate losses from being outperformed.

\subsubsection{Some additional special cases}
A few other special cases are worth discussing. If $\sigma =0$ a.s., there is no common noise. In this case, $\varphi =\psi =0,$ and in turn the MFE is equal to the Merton portfolio, which means that the agents are not at all competitive.
On the other hand, if $\nu =0$ a.s., there is no independent noise. In this case, $\varphi=\E[\delta\mu/\sigma]$ and $\psi=\E[\theta(\delta-1)]$, and the optimal portfolio becomes
\[
\pi^* = \delta\frac{\mu}{\sigma^2} - \frac{\theta(\delta-1)}{\sigma(1+\E[\theta(\delta-1)])}\E\left[\delta\frac{\mu}{\sigma}\right].
\]
Lastly, if all agents have the same type vector (i.e., $\zeta$ is deterministic), then $\pi^*$ is deterministic and, furthermore,
\[
\pi^* = \frac{\delta\mu}{(1+\theta(\delta-1))\sigma^2+\nu^2}.
\]

\section{Conclusions and extensions} \label{se:conclusions}
We have considered optimal portfolio management problems under relative
performance concerns for both finite and infinite populations. The agents have a common investment horizon and either CARA or CRRA risk preferences, and they trade individual stocks with log-normal dynamics driven by both common and idiosyncratic noises. They face competition in that their individual utility criteria depend on both their
individual wealth as well as the wealth of the others. We have explicitly constructed the
associated constant Nash and mean field equilibria.

Our study points to several directions for future research. A first direction is to further analyze the finite population problems by using concepts from indifference valuation. Indeed, as we
mentioned in the proof of Theorem \ref{th:exp-eq-np}, we may identify the effect of competition as a liability and, in turn, solve an indifference valuation
problem.
Similarly, for the CRRA case, one may relate the
competition to a multiplicative liability factor.
There is a fundamental difference, of course, between the classical
indifference pricing problems and the ones herein; namely, the liability is
essentially endogenous, as it depends on the actions of the agents.
Nevertheless, employing indifference valuation arguments is expected to
yield a clearer financial interpretation of the equilibrium strategies by relating
them to indifference hedging strategies. It will also permit an analysis of
sensitivity effects of varying agents' population size using arguments
from the so-called relative indifference valuation. 
Such questions are left for future work.

Herein, the fund managers are only concerned with maximization of utility coming from terminal wealth (both absolute and relative to other agents), but one could also incorporate intermediate consumption. There are two natural ways to do this, discussed only for the CRRA model for concreteness.
First, one might add a utility-of-consumption term to the optimization criterion \eqref{def:Jpower} and modify the individual wealth process to 
\[
dX^i_t = \pi^i_tX^i_t(\mu_idt + \nu_idW^i_t + \sigma_idB_t) - C^i_tdt.
\]
While the calculations may be tedious, we expect that this problem is tractable. 
A more interesting approach would build on the first by incorporating \emph{relative consumption standards}, modeled as constraints in the form of lower bounds on the rate of consumption, which could themselves depend on the consumption of other agents.
This setting would reflect the more realistic situation in which individual consumption standards are affected by the behavior of other agents.

An important assumption of the model herein is that each agent has full information of the individual preference and market parameters of each other agent. This is also the main modeling
ingredient in \cite{basak2015competition} and is partially defended by the
fact that fund managers post their returns publicly, and from these announcements certain
information can in turn be inferred by their competitors. While this is
undoubtedly a considerable modeling limitation, our results give new
solutions to existing problems, especially for mean field games with
non-quadratic criteria. Furthermore, this assumption of common knowledge
may be relaxed if one introduces ambiguity around the publicly posted
competitors' parameters. This ambiguity may be, for example, modeled through
an error margin depending on individual views. This would give rise to a
class of interesting mean field games with filtering.

In a different direction, a natural generalization of our model would allow agents to invest in any of the stocks, not just the individual stock assigned to them. Such a case has been recently analyzed in \cite{basak2015competition}, and in \cite{anthropelos-geng-zariphopoulou} under forward performance criteria. Important questions arise on the effects of competition to asset specialization. While such generalization
might be intractable for the finite population setting, a mean field formulation
may provide a more tractable framework for studying the interactive role of
competition and asset familiarity, specialization and competition.

Finally, one may extend the current model to dynamically evolving markets
and rolling horizons. Such generalization may be analyzed under forward
performance criteria, extending the results of \cite{anthropelos-geng-zariphopoulou}, 
and would lead naturally to a new class of mean field games. It would also allow for further extending the concept of benchmarking under forward criteria, introduced in \cite{musiela2009portfolio}.

\subsection*{Acknowledgements}
The authors are grateful to Michalis Anthropelos, Mihai Sirbu, and especially Gon\c{c}alo dos Reis for their helpful comments. 
This work was presented at the 8$^\text{th}$ Western Conference on Mathematical Finance in Seattle in 2017; the International Workshop on SPDE, BSDE, and their Applications in Edinburgh in 2017; and the Conference on Kinetic Theory in Austin in 2017. The authors would like to thank the participants for fruitful comments and suggestions.

\bibliographystyle{amsplain}
\bibliography{Investment-MFG}

\end{document}